\documentclass[%
 reprint,
superscriptaddress,
 amsmath,amssymb,
 aps,
pra,
]{revtex4-2}

\usepackage{amsmath,amsthm,amssymb,amsfonts,amsbsy,times}
\usepackage{graphicx}
\usepackage{textcomp}
\usepackage{xcolor}
\usepackage{subfigure}
\usepackage{epsfig}
\usepackage{multirow}
\usepackage{algorithm}
\usepackage{algorithmicx}
\usepackage{algpseudocode}
\usepackage{array}
\usepackage{epstopdf}
\usepackage{color}
\usepackage{diagbox}
\usepackage{bm}
\usepackage{bbm}
\usepackage{mathrsfs}
\usepackage{tcolorbox}
\usepackage{pdfsync}
\usepackage{mathtools}

\usepackage{url}
\usepackage[hidelinks]{hyperref}
\usepackage{paralist}
\usepackage{comment}
\usepackage{multirow}
\usepackage[toc, page]{appendix}

\usepackage{fancyhdr}
\usepackage[capitalize]{cleveref}

\DeclareMathOperator{\Var}{Var}
\DeclareMathOperator{\Cov}{Cov}

\crefrangelabelformat{figure}{#3#1#4--#5#2#6}
\crefrangelabelformat{equation}{(#3#1#4--#5#2#6)}
\crefmultiformat{equation}{Eqs.~(#2#1#3}{,~#2#1#3)}{,~#2#1#3}{,~#2#1#3)}

\newtheorem{lemma}{Lemma}
\newtheorem{cor}{Corollary}
\newtheorem{prop}{Proposition}

\newtheorem{theorem}{Theorem}

\theoremstyle{remark}

\theoremstyle{problem}

\newcommand{\R}{\mathbb{R}}

\newcommand{\C}{\mathbb{C}}

\newcommand{\e}{\begin{equation}}
\newcommand{\ee}{\end{equation}}
\newcommand{\en}{\begin{equation*}}
\newcommand{\een}{\end{equation*}}
\newcommand{\eqn}{\begin{eqnarray}}
\newcommand{\eeqn}{\end{eqnarray}}
\newcommand{\bmat}{\begin{bmatrix}}
\newcommand{\emat}{\end{bmatrix}}

\DeclareMathAlphabet\mathbfcal{OMS}{cmsy}{b}{n}

\newcommand{\E}{\operatorname{\mathbb{E}}}

\newcommand{\vct}[1]{\boldsymbol{#1}}
\newcommand{\mtx}[1]{\boldsymbol{#1}}

\newcommand{\<}{\langle}
\renewcommand{\>}{\rangle}

\newcommand{\trace}{\operatorname{trace}}

\newcommand{\wt}{\widetilde}

\newcommand{\calF}{\mathcal{F}}

\newcommand{\vg}{\vct{g}}

\newcommand{\vu}{\vct{u}}

\newcommand{\vw}{\vct{w}}
\newcommand{\vx}{\vct{x}}
\newcommand{\vy}{\vct{y}}

\newcommand{\vtheta}{\vct{\theta}}

\newcommand{\vphi}{\vct{\phi}}

\newcommand{\vrho}{\vct{\rho}}

\newcommand{\mG}{\mtx{G}}
\newcommand{\mH}{\mtx{H}}

\newcommand{\mL}{\mtx{L}}

\newcommand{\mP}{\mtx{P}}

\newcommand{\mR}{\mtx{R}}

\newcommand{\mU}{\mtx{U}}

\newcommand{\mX}{\mtx{X}}
\newcommand{\mY}{\mtx{Y}}
\newcommand{\mZ}{\mtx{Z}}

\newcommand{\mDelta}{\mtx{\Delta}}

\newcommand{\mId}{{\bf I}}

\graphicspath{{./figs/}}

\newlength{\imgwidth}
\newboolean{twoColVersion}
\setboolean{twoColVersion}{false}
\newcommand{\twoCol}[2]{\ifthenelse{\boolean{twoColVersion}} {#1} {#2} }

\definecolor{teal}{RGB}{26,157,150}

\usepackage{mathtools}

\begin{document}

\preprint{APS/123-QED}

\title{From Barren Plateaus to SPSA Optimization \\ in Variational Quantum Eigensolvers}

\author{Zhen Qin}
\email{zhenqin@umich.edu}
\affiliation{Michigan Institute for Computational Discovery and Engineering, Department of Electrical Engineering and Computer Science, Department of Statistics, University of Michigan, Ann Arbor, MI 48109 USA}

\begin{abstract}
The barren plateau (BP) phenomenon poses a fundamental challenge to the trainability of variational quantum eigensolvers (VQEs) by causing exponentially vanishing gradients as the system size increases. While extensive studies have investigated the geometric origins of BP, its impact on the optimization dynamics and complexity of practical algorithms under finite-shot measurements has received comparatively less attention. In this paper, we develop a theoretical framework that characterizes how the BP affects the optimization dynamics of the Simultaneous Perturbation Stochastic Approximation (SPSA) algorithm and quantifies the resulting iteration complexity and measurement budget. We derive non-asymptotic bias and variance bounds for the SPSA gradient estimator, introduce a signal-to-noise ratio (SNR) defined as the ratio of the gradient energy to the mean squared error of the estimator, and establish convergence guarantees for SPSA under finite-shot measurements. Our results show that, under the proposed parameter schedules, the number of iterations sufficient to guarantee a fixed relative first-order stationarity level scales polynomially with the inverse of the intrinsic gradient energy. Consequently, when this energy decays exponentially under BP, the resulting guaranteed iteration complexity and total measurement budget both grow exponentially with the number of qubits.
\end{abstract}

\maketitle

\section{Introduction}
\label{sec:introduction}

The Variational Quantum Eigensolver (VQE) \cite{peruzzo2014variational,mcclean2016theory,kandala2017hardware,wang2019accelerated,tilly2022variational} is one of the leading variational algorithms for solving large-scale eigenvalue problems on Noisy Intermediate-Scale Quantum (NISQ) devices. Its practical success, however, critically depends on the ability to efficiently optimize parameterized quantum circuits, a task that remains one of the major challenges in variational quantum computing. A fundamental obstacle is the barren plateau (BP) phenomenon \cite{wecker2015progress,mcclean2018barren,bittel2021training,cerezo2021cost,larocca2025barren,qin2026geometric}, under which the objective landscape becomes increasingly flat as the system size grows. More precisely, the gradients of the objective function decay exponentially with the number of qubits, so that local perturbations of the circuit parameters induce only exponentially small changes in the objective value. Consequently, optimization algorithms receive progressively weaker local information, substantially degrading their ability to identify effective descent directions and rendering the training of deep parameterized quantum circuits increasingly difficult. Extensive theoretical studies have identified several mechanisms contributing to the emergence of BP, including the exponentially large Hilbert space \cite{volkoff2021large,diaz2023showcasing,ragone2024lie,fontana2024characterizing,cerezo2025does}, circuit architecture and depth \cite{mcclean2018barren,cerezo2021cost,larocca2022diagnosing,holmes2022connecting}, parameter initialization strategies \cite{grant2019initialization,patti2021entanglement,sauvage2021flip,rad2022surviving,zhang2022escaping,wang2024trainability,puig2025variational}, and hardware noise \cite{wang2021noise,stilck2021limitations,de2023limitations,schumann2024emergence,sannia2024engineered,liu2025stochastic,singkanipa2025beyond}.

Despite the substantial progress in understanding the mechanisms underlying BP, its consequences for optimization algorithms have received comparatively less attention. It was shown in \cite{arrasmith2021effect} that BP also affects gradient-free optimizers: because cost-function differences between parameter points are exponentially concentrated, progress can only be made when the cost is resolved to exponential precision.
In a complementary direction, non-asymptotic iteration-complexity guarantees were established for zeroth-order methods under objective evaluations corrupted by both zero-mean and biased noise~\cite{kungurtsev2024iteration}.
These two lines of work address different aspects of the problem. The former characterizes the precision required by a BP landscape, but does not analyze the dynamics of a specific optimizer or the accumulation of statistical errors over iterations. The latter, in contrast, provides iteration-complexity guarantees under noisy evaluations, but its bounds depend on generic noise and smoothness constants rather than on the gradient scale of the underlying landscape.
 In practical VQE implementations, the information available to the optimizer is estimated from finite-shot measurements~\cite{tilly2022variational,bittel2022fast,teo2023optimized,bonet2023performance,scriva2024challenges,kaminishi2026impact}. As a result, each optimization step is affected by statistical fluctuations whose magnitude is determined by the measurement budget, while the useful optimization signal is governed by the underlying landscape. The interplay between these two scales is particularly important in BP landscapes, where the optimization signal can become exponentially small. Yet, how this interplay affects the dynamics and complexity of a practical optimizer has not been rigorously quantified. This motivates the following question: \emph{how does BP affect the optimization dynamics and complexity of SPSA when the measurement noise is explicitly tied to the shot budget and the gradient signal is governed by the BP scaling?}

Among the optimization methods developed for VQE, the Simultaneous Perturbation Stochastic Approximation (SPSA) algorithm \cite{kandala2017hardware,ganzhorn2019gate,cade2020strategies,bonet2023performance,grossi2023finite,tran2024variational,jiang2024error,kungurtsev2024iteration,le2026variational} has become one of the most widely adopted approaches because of its exceptional measurement efficiency. Unlike conventional gradient-based methods, such as finite-difference method \cite{uvarov2020variational}, the parameter-shift rule \cite{schuld2019evaluating}, and quantum natural gradient method \cite{stokes2020quantum}, whose computational cost generally increases with the number of optimization parameters, SPSA estimates the gradient using only two objective-function evaluations per iteration, independent of the problem dimension. Under finite-shot measurements, this translates into a measurement cost that is also independent of the number of optimization parameters, assuming a fixed Hamiltonian and a fixed number of shots per Pauli term. This dimension-independent measurement cost makes SPSA particularly attractive for large-scale variational quantum circuits, where measurement resources are often the primary computational bottleneck. As a result, SPSA has been widely adopted in practical VQE implementations on current quantum hardware.

Motivated by these observations, in this paper we develop a theoretical framework to characterize how the BP affects SPSA optimization under finite-shot measurements. Specifically, we first derive non-asymptotic bounds on the bias and variance of the SPSA gradient estimator under finite-shot measurements, explicitly quantifying how these statistical properties depend on both the variational circuit parameters and the allocated measurement budget. Building upon this characterization, we introduce a signal-to-noise ratio (SNR), defined as the ratio of the gradient energy to the mean squared error of the SPSA gradient estimator, and derive a measurement budget sufficient to achieve a prescribed SNR level. Finally, we extend this analysis to the full optimization trajectory and develop a convergence theory for SPSA under finite-shot measurements, deriving sufficient conditions on both the number of iterations and the total measurement budget for reaching a prescribed level of relative first-order stationarity. We further show that the exponentially decaying gradient energy associated with BP leads to exponentially growing sufficient bounds on both the number of iterations and the total measurement budget.

\section{The SPSA Algorithm: Statistical Properties and Convergence Guarantee}

We begin by introducing the SPSA framework for VQEs under measurement noise. Let
$f(\vtheta)$ denote the ideal objective function, which corresponds to the expectation value of the target Hamiltonian with respect to the parameterized quantum state. In practice, however, this expectation can only be estimated from finite-shot measurements, leading to the noisy objective
\begin{eqnarray}
\label{noisy loss VQE}
\hat f(\vtheta) &\!\!=\!\!& f(\vtheta)+\xi(\vtheta)\nonumber\\
&\!\!:=\!\!& \langle \vphi_0 | \mU_1^\dagger(\theta_1)\cdots \mU_N^\dagger(\theta_N) \mH \mU_N(\theta_N)\nonumber\\
&\!\!\!\!& \cdots \mU_1(\theta_1) |\vphi_0\rangle + \xi(\vtheta),
\end{eqnarray}
where $\vtheta=\begin{bmatrix}\theta_1 & \cdots & \theta_N\end{bmatrix}^\top\in\R^N$ is the variational parameter vector, $|\vphi_0\rangle$ is the reference state, $\mU_1(\theta_1),\ldots,\mU_N(\theta_N)$ are parameterized unitary gates, each depending on a single parameter, so that $N$ equals the number of variational parameters, $\mH=\sum_{\alpha=1}^{L}w_\alpha\mP_\alpha$ is the target Hamiltonian with Pauli operators $\mP_\alpha\in\{\mId,\mX,\mY,\mZ\}^{\otimes n}$ satisfying $\|\mP_\alpha\|=1$, and $\xi(\vtheta)$ models the stochastic measurement error arising from finite-shot sampling. A layered ansatz with several parameterized gates per layer is described by treating each parameterized gate as a separate $\mU_\ell$, and parameter-free gates such as CNOTs can be absorbed into the parameterized gates by conjugation.

To minimize the noisy objective, SPSA updates the variational parameters according to
\begin{eqnarray}
\label{SPSA iteration}
\vtheta_{t+1} = \vtheta_t - \mu_t \hat \vg(\vtheta_t),
\end{eqnarray}
where $\mu_t>0$ denotes the step size and $\hat\vg(\vtheta_t)=[\hat g_1(\vtheta_t), $ $\ldots,\hat g_N(\vtheta_t)]^\top\in\R^N$ is the stochastic gradient estimator. Each SPSA iteration estimates the gradient using two objective-function evaluations, each of which estimates the expectation values of all $L$ Pauli terms in $\mH$ using $M$ measurement shots per term. Consequently, the total measurement cost per iteration is $2LM$ shots, independent of the dimension $N$.
 Specifically, each component of the estimator is given by
\begin{eqnarray}
\label{elements in the gradient of SPSA iteration}
 \hat g_{\ell}(\vtheta_t)  = \frac{\hat f(\vtheta_t + c_t\mDelta_t) - \hat f(\vtheta_t- c_t\mDelta_t) }{2c_t\Delta_{\ell,t}},
\end{eqnarray}
where $c_t>0$ is the perturbation radius and $\mDelta_t=[\Delta_{1,t},  $ $\ldots,\Delta_{N,t}]^\top$ is a random perturbation vector whose entries are independent Rademacher random variables, i.e., $\mathbb{P}(\Delta_{\ell,t}=1)=\mathbb{P}(\Delta_{\ell,t}=-1)=1/2$.

The remainder of this section establishes the statistical properties of the SPSA gradient estimator and presents the convergence guarantee for the resulting stochastic optimization algorithm.

\subsection{Statistical Properties} 

We first fix an iteration $t$ and establish several fundamental properties of the VQE objective, which will serve as the basis for the subsequent statistical analysis of the SPSA gradient estimator. Throughout this section, without loss of generality, we assume that each parameterized quantum gate takes the form $ \mU_\ell(\theta_\ell)=e^{-i\theta_\ell\mG_\ell}$, where $\mG_\ell\in\C^{2^n\times 2^n}$ denotes the Hermitian generator of the $\ell$-th parameterized gate and satisfies $\|\mG_\ell\|\le1$ for $\ell=1,\ldots,N$. Since conjugation by a parameter-free unitary preserves the spectral norm, absorbing such gates does not affect this condition.
\begin{prop}
\label{Proposition: Hessian Property}
Let $f(\vtheta)=\langle \vphi_0|\mU^\dagger(\vtheta)\mH\mU(\vtheta)|\vphi_0\rangle$ where $\mH = \sum_{\alpha=1}^L w_\alpha \mP_\alpha$, and  $\mU(\vtheta)=\mU_N(\theta_N)\cdots \mU_1(\theta_1)$. Suppose that $\mU_\ell(\theta_\ell)=e^{-i\theta_\ell \mG_\ell}$ with $\|\mG_\ell\|\le 1$ for all $\ell$. Then the Hessian of $f$ is Lipschitz continuous with respect to the spectral norm:
\begin{eqnarray}
\label{Hessian Lipschitz for spectral norm}
&&\|\nabla^2f(\vx)-\nabla^2f(\vy)\|  \nonumber\\
&\le& 8\,\|\vw\|_1\,N^{\frac{3}{2}}\|\vx-\vy\|_2, \  \forall \vx,\vy\in\mR^{N},
\end{eqnarray}
where $\vw = \begin{bmatrix} w_1 & \cdots & w_L \end{bmatrix}^\top\in\R^{L}$.
\end{prop}
The proof is deferred to Appendix~\ref{Proof of Hessian property}. Proposition~\ref{Proposition: Hessian Property} provides an explicit Hessian Lipschitz constant for VQE objectives, which serves as a key ingredient in the subsequent analysis of the SPSA estimator. In addition to this intrinsic geometric effect, practical VQE implementations estimate objective values through finite-shot measurements, which introduces stochastic fluctuations. The following lemma establishes a measurement noise model for the objective-value estimation and characterizes the corresponding statistical properties.

\begin{lemma}
\label{lem:noise-model}
Let $f(\vtheta)=\langle \vphi_0|\mU^\dagger(\vtheta)\mH\mU(\vtheta)|\vphi_0\rangle$ where $\mH = \sum_{\alpha=1}^L w_\alpha \mP_\alpha$, and  $\mU(\vtheta)=\mU_N(\theta_N)\cdots \mU_1(\theta_1)$. Fix a perturbation direction $\mDelta$ and define $\vtheta^\pm=\vtheta\pm c\mDelta$. For each $s\in\{+,-\}$, the corresponding quantum state is $\vrho(\vtheta^s) = \mU(\vtheta^s) |\vphi_0\rangle \langle\vphi_0| \mU^\dagger(\vtheta^s)$. For each Pauli observable $\mP_\alpha$ with spectral decomposition $\mP_\alpha={\bm \Pi}_\alpha^+-{\bm \Pi}_\alpha^-$, denote by $X_{\alpha,k}^s\in\{\pm1\}$, $k=1,\ldots,M$, the $M$ i.i.d. measurement outcomes, where the outcomes $\pm1$ are associated with the projectors ${\bm \Pi}_\alpha^\pm$, respectively. The distribution of these random variables is given by $\Pr  ( X_{\alpha,k}^s=\pm1 \, |\, \mDelta ) = \trace\!\left( \vrho(\vtheta^s){\bm \Pi}_\alpha^\pm \right)$. The empirical estimator of the Pauli expectation value is obtained by averaging the corresponding measurement outcomes, $\hat P_\alpha^s = \frac1M\sum_{k=1}^{M}X_{\alpha,k}^s$, and the oracle estimate is constructed as $\hat f^s = \sum_{\alpha=1}^{L}w_\alpha\hat P_\alpha^s$. All measurement outcomes $\{X_{\alpha,k}^s\}$ are mutually independent conditional on $\mDelta$ for different choices of the Pauli index $\alpha$, the perturbation sign $s$, and the shot index $k$. Accordingly, define the measurement noise by
\begin{eqnarray}
\label{oracle noise}
\xi(\vtheta^s) := \hat f^s-f(\vtheta^s).
\end{eqnarray}
Conditioned on $\mDelta$, this noise satisfies
\begin{eqnarray}
\label{properties of noise1}
\E[\xi(\vtheta^s)\mid\mDelta]&\!\! =\!\! &0,\\
\label{properties of noise2}
\Var(\xi(\vtheta^s)\mid\mDelta) &\!\! \le\!\! & \frac{\|\vw\|_2^2}{M},
\end{eqnarray}
with $\vw = \begin{bmatrix} w_1 & \cdots & w_L \end{bmatrix}^\top\in\R^{L}$ and the noises associated with the two perturbations are independent, i.e.,
\begin{eqnarray}
\label{independence of noise}
\xi(\vtheta^+) \perp \xi(\vtheta^-) \mid \mDelta.
\end{eqnarray}
\end{lemma}
The proof is provided in {Appendix}~\ref{proof of lemma different noise properties}. The preceding lemma isolates the contribution of finite-shot measurements from the intrinsic geometry of the VQE objective. In particular, the measurement noise is conditionally unbiased, its variance decreases inversely with the  measurement budget, and the two perturbed objective evaluations are conditionally independent. These properties, together with the Hessian Lipschitz continuity established in Proposition~\ref{Proposition: Hessian Property}, provide the foundation for analyzing the statistical behavior of the SPSA gradient estimator.

\begin{theorem}\label{thm:main property of noisy gradient}
Consider the VQE objective $f(\vtheta) = \langle \vphi_0| \mU^\dagger(\vtheta) $ $\mH\mU(\vtheta) |\vphi_0\rangle$ with Hamiltonian  $\mH=\sum_{\alpha=1}^{L}w_\alpha\mP_\alpha$, and denote  $\vw = \begin{bmatrix} w_1 & \cdots & w_L \end{bmatrix}^\top\in\R^{L}$ Suppose that each parameterized quantum gate is of the form $\mU_\ell(\theta_\ell)=e^{-i\theta_\ell \mG_\ell}$, where $\mG_\ell\in\C^{2^n\times 2^n}$ denotes the Hermitian generator of the $\ell$-th parameterized gate and satisfies $\|\mG_\ell\|\le1$ for $\ell=1,\ldots,N$. Under the measurement noise model in Lemma~\ref{lem:noise-model}, the SPSA gradient estimator satisfies, for every $\ell=1,\ldots,N$,
\begin{eqnarray}
\label{bias of noisy gradient}
\big|\E[\hat g_\ell(\vtheta)] - \partial_\ell f(\vtheta)\big| \;\le\; 4c^2 N^3\|\vw\|_1,
\end{eqnarray}
where $\partial_\ell f$ denotes the partial derivative of $f(\vtheta)$ with respect to $\theta_\ell$. Moreover,
\begin{eqnarray}
\label{variance of noisy gradient}
&\!\!\!\!&\Var(\hat g_\ell(\vtheta) - \partial_\ell f(\vtheta)) \nonumber\\
&\!\!\le\!\!&  \|\nabla f(\vtheta)\|_2^2 - (\partial_{\ell} f(\vtheta))^2 + \frac{\|\vw\|_2^2}{2c^2M}+ 16  c^4 N^6 \|\vw\|_1^2 \nonumber\\
&\!\!\!\!& +  8c^2 N^{3} \| \vw\|_1 \sqrt{\|\nabla f(\vtheta)\|_2^2 - (\partial_{\ell} f(\vtheta))^2},
\end{eqnarray}
where $\nabla f(\vtheta) = \begin{bmatrix} \partial_1 f(\vtheta) & \cdots & \partial_N f(\vtheta)  \end{bmatrix}^\top\in\R^{N}$. Here, the expectation and variance are taken with respect to the joint randomness of the SPSA perturbation $\mDelta$ and the measurement noise $\xi(\vtheta)$.
\end{theorem}

The proof is provided in Appendix~\ref{Proof of theorem main property of noisy gradient}. The proof is deferred to Appendix~\ref{Proof of theorem main property of noisy gradient}. Theorem~\ref{thm:main property of noisy gradient} reveals that the statistical behavior of the SPSA gradient estimator is governed jointly by the local geometry of the VQE objective and the finite-shot measurement process. In particular, the variance bound naturally decomposes into three components: the intrinsic variance induced by the random perturbation of SPSA, the measurement variance that decreases with the measurement budget $M$, and a higher-order term arising from the Hessian Lipschitz continuity of the objective function.

The variance in \eqref{variance of noisy gradient} alone does not determine how close $\hat g_\ell(\vtheta)$ is to $\partial_\ell f(\vtheta)$, since the estimator is biased according to \eqref{bias of noisy gradient}, and this bias becomes significant relative to the gradient signal when the gradient itself is small. To account for both effects, we quantify the estimation error through the mean squared error (MSE)
\begin{eqnarray}
\label{eq:MSE}
\mathrm{MSE}_\ell(\vtheta) &\!\!:=\!\!& \E\big[(\hat g_\ell(\vtheta)-\partial_\ell f(\vtheta))^2\big] \nonumber\\
&\!\!=\!\!& \Var(\hat g_\ell(\vtheta)-\partial_\ell f(\vtheta))\nonumber\\
&\!\!\!\!&+ \big(\E[\hat g_\ell(\vtheta)]-\partial_\ell f(\vtheta)\big)^2,
\end{eqnarray}
which accounts for both the bias and the variance. To quantify the reliability of the estimated gradient, we compare the strength of the true optimization signal with this estimation error. Specifically, because the gradient component $\partial_\ell f(\vtheta)$ may change sign over the parameter landscape, its first moment does not provide a meaningful measure of its magnitude. Instead, we characterize the signal strength through its second moment. This motivates the following gradient SNR:
\begin{eqnarray}
\label{definition of population SNR}
{\mathrm{SNR}}_\ell := \frac{ \E_{\Theta} [(\partial_\ell f(\vtheta))^2] }{ \E_{\Theta} [\mathrm{MSE}_\ell(\vtheta)] },
\end{eqnarray}
where the expectation $\E_{\Theta} [ \cdot ]$ is taken over the parameter distribution of $\vtheta$. Because the denominator includes the squared bias, a larger $\mathrm{SNR}_\ell$ implies a smaller mean-square deviation of $\hat g_\ell(\vtheta)$ from $\partial_\ell f(\vtheta)$ relative to the gradient signal, on average over the parameter distribution.

The per-component SNR of SPSA is, however, intrinsically limited. Since all components of the gradient estimator are obtained from the same random perturbation, each component $\hat g_\ell(\vtheta)$ is contaminated by the other $N-1$ partial derivatives, which contributes an error of variance $\|\nabla f(\vtheta)\|_2^2-(\partial_\ell f(\vtheta))^2$ regardless of $c$ and $M$. As a result, even with exact objective evaluations and a vanishing perturbation, $\mathrm{SNR}_\ell$ cannot exceed approximately $1/(N-1)$ when all gradient components have comparable magnitudes. This limitation is inherent to SPSA, which recovers all $N$ gradient components from a single directional difference per iteration. Therefore, the SNR of an individual SPSA estimate should be evaluated relative to this intrinsic level rather than to an absolute threshold.

\begin{cor}
\label{cor:sample complexity barren plateau}
Under the setting of \Cref{thm:main property of noisy gradient}, suppose $\E_{\Theta}[\|\nabla f(\vtheta)\|_2^2] = N\,\E_{\Theta}[(\partial_\ell f(\vtheta))^2]$, and choose the perturbation parameter as $c^2 = \gamma\sqrt{\E_{\Theta}[(\partial_\ell f(\vtheta))^2]}$ for some $\gamma>0$. Define
\begin{eqnarray}
\label{eq:A def}
A(N,\vw,\gamma) := \big(\sqrt{N-1}+4\sqrt{2}\,\gamma N^3\|\vw\|_1\big)^2 .
\end{eqnarray}
Then, for any $\varepsilon_1\in(0,1)$, $\mathrm{SNR}_\ell \ge \varepsilon_1/A(N,\vw,\gamma)$ holds whenever
\begin{eqnarray}
\label{cor:sample complexity lower bound c scaled}
M \ \ge\ \frac{\varepsilon_1}{1-\varepsilon_1}\cdot\frac{\|\vw\|_2^2}{2\gamma\,A(N,\vw,\gamma)\,\big(\E_{\Theta}[(\partial_\ell f(\vtheta))^2]\big)^{3/2}} .
\end{eqnarray}
\end{cor}

The proof is provided in {Appendix}~\ref{Proof of SNR lower bound requirement}. Here, the target level $\varepsilon_1/A(N,\vw,\gamma)$ is a fraction $\varepsilon_1$ of $1/A(N,\vw,\gamma)$, which does not exceed the intrinsic level $1/(N-1)$; the closer $\varepsilon_1$ is to one, the larger the measurement budget required by \eqref{cor:sample complexity lower bound c scaled}. The exponent $3/2$ in \eqref{cor:sample complexity lower bound c scaled}, rather than the exponent $1$ expected from shot noise alone, originates from the perturbation parameter $c$. Since the squared finite-difference bias grows with $c$ while the measurement-noise term $\|\vw\|_2^2/(2c^2M)$ decreases with $c$, the two terms are balanced by choosing $c^2=\Theta\big(\sqrt{\E_{\Theta}[(\partial_\ell f(\vtheta))^2]}\big)$. Thus, $c$ is chosen to shrink as the landscape flattens. This shrinkage inflates the measurement-noise term by an additional factor of $\big(\E_{\Theta}[(\partial_\ell f(\vtheta))^2]\big)^{-1/2}$, raising the exponent from $1$ to $3/2$.

In particular, we choose $\gamma=\gamma_0/(N^3\|\vw\|_1)$ with a constant $\gamma_0>0$, which ensures the following two properties. First, substituting $c^2=\gamma\sqrt{\E_{\Theta}[(\partial_\ell f(\vtheta))^2]}$ into \eqref{bias of noisy gradient} yields
\begin{eqnarray}
\label{eq:relative bias}
\big|\E[\hat g_\ell(\vtheta)]-\partial_\ell f(\vtheta)\big| \;\le\; 4\gamma_0\sqrt{\E_{\Theta}[(\partial_\ell f(\vtheta))^2]},
\end{eqnarray}
so that the bias remains a fixed fraction of the typical gradient magnitude, uniformly in $N$ and $\vw$. Second, substituting the same choice into \eqref{eq:A def} yields
\begin{eqnarray}
\label{eq:A scaled}
A(N,\vw,\gamma) \;=\; \big(\sqrt{N-1}+4\sqrt{2}\,\gamma_0\big)^2 \;=\; \Theta(N),
\end{eqnarray}
so that the target level $\varepsilon_1/A(N,\vw,\gamma)$ remains within a constant factor of the intrinsic level $1/(N-1)$. With this choice, \eqref{cor:sample complexity lower bound c scaled} becomes
\begin{eqnarray}
\label{eq:M scaled}
M  &\!\!\ge\!\!& \frac{\varepsilon_1}{1-\varepsilon_1}\cdot\frac{N^3\|\vw\|_1\|\vw\|_2^2}{2\gamma_0\,A(N,\vw,\gamma)\,\big(\E_{\Theta}[(\partial_\ell f(\vtheta))^2]\big)^{3/2}}\nonumber\\
&\!\!=\!\!& \frac{\varepsilon_1}{1-\varepsilon_1}\cdot\Theta\bigg(\frac{N^2\|\vw\|_1\|\vw\|_2^2}{\big(\E_{\Theta}[(\partial_\ell f(\vtheta))^2]\big)^{3/2}}\bigg),
\end{eqnarray}
where $\gamma_0$ is treated as a constant in the $\Theta(\cdot)$ notation. This expression shows that the sufficient measurement budget grows only polynomially with $N$ and $\vw$, while its scaling with the gradient energy is governed by the exponent $3/2$.

We next examine how the BP phenomenon affects this requirement. Following~\cite{larocca2025barren}, we formulate the BP condition with respect to a distribution over the circuit parameters. Specifically, a cost function is said to exhibit a BP if
$\Var_{\Theta}[\partial_\ell f(\vtheta)]=O(\beta^{-n})$ for some $\beta>1$ \cite[Eq.~(5)]{larocca2025barren},
where the base $\beta$ depends on the circuit, the initial state, and the observable. A standard setting in which this condition arises, and which we adopt throughout, is that the parameters $\theta_\ell$ are drawn independently and uniformly from $[-\pi,\pi]$, while the circuit is sufficiently deep for the induced distribution of $\mU(\vtheta)$ to form an approximate unitary 2-design over its dynamical Lie group. In addition, when each generator has the form $\mG_\ell=\mP_\ell/2$ with $\mP_\ell$ a Pauli operator, $f(\vtheta)$ is $2\pi$-periodic in each $\theta_\ell$. Consequently, under the above parameter distribution,
$\E_{\Theta}[\partial_\ell f(\vtheta)]=0$,
since the integral of a derivative over one full period vanishes.
 Under this property, the gradient variance coincides with its second moment, yielding $\E_{\Theta}[(\partial_\ell f(\vtheta))^2]=O(\beta^{-n})$, where $n$ denotes the number of qubits. Consequently, as long as $N$, $\|\vw\|_1$, and $\|\vw\|_2^2$ grow at most polynomially in $n$, as is typical for local Hamiltonians and polynomial-depth circuits, the sufficient measurement budget in \eqref{eq:M scaled} grows as $\beta^{3n/2}$ up to polynomial factors. Therefore, BPs not only slow down optimization due to vanishing gradients, but also increase the statistical cost of obtaining reliable gradient estimates within our guarantees.

\subsection{Convergence Guarantee}  
The SNR analysis above characterizes the statistical reliability of the SPSA gradient estimator at each iteration. However, a sufficiently accurate gradient estimate alone does not directly characterize the optimization performance, since the stochastic gradient errors accumulate throughout the iterative optimization process. To establish a complete understanding of SPSA in VQE optimization, we next analyze the convergence behavior of the resulting optimization trajectory and quantify how the estimator accuracy affects the number of iterations and measurement resources required to approach a stationary region.

\begin{theorem}
\label{thm: upper bound of the gradient energy}
Consider the VQE objective  $f(\vtheta) = \langle \vphi_0| \mU^\dagger(\vtheta)$ $\mH\mU(\vtheta) |\vphi_0\rangle$, where the Hamiltonian admits the decomposition $\mH=\sum_{\alpha=1}^{L}w_\alpha\mP_\alpha$ and let $\vw = \begin{bmatrix} w_1 & \cdots & w_L \end{bmatrix}^\top\in\R^{L}$.
Under the measurement noise model in Lemma~\ref{lem:noise-model}, we consider the following parameters for the SPSA estimator in \eqref{SPSA iteration}: $\mu_t=\mu_0T^{-1/2}$, $c_t=c_0T^{-1/8}$ and $M_t=M_0T^{1/4}$, where $\mu_0\le \frac{1}{4(3N^2+N)\|\vw\|_1}$, $c_0>0$ and $M_0>0$.
Then, after $T$ iterations, the minimum expected gradient energy among the generated iterates satisfies
\begin{eqnarray}
\label{expansion of iteration2 another4 upper bound energy2 main paper}
\min_{0\leq t \leq T-1}\E[\|\nabla f(\vtheta_t)\|_2^2] \leq \kappa T^{-\frac{1}{2}},
\end{eqnarray}
establishing the $O(T^{-\frac12})$ convergence rate. Here $\kappa \!=\! \frac{8 \|\vw\|_1}{\mu_0} + \frac{4N^2\|\vw\|_1\|\vw\|_2^2 \mu_0}{c_0^2 M_0} + 72 N^{7}\|\vw\|_1^2 c_0^4$. To achieve a relative gradient energy level $\epsilon_2>0$, i.e., $\min_{0\leq t \leq T-1}\E[\|\nabla f(\vtheta_t)\|_2^2] \leq \epsilon_2\, \E_{\Theta}[\|\nabla f(\vtheta)\|_2^2]$, it suffices to choose
\begin{eqnarray}
\label{the necessary condition for the number of iterations main paper}
T \;\ge\; \frac{\kappa^2}{\epsilon_2^2\, \big(\E_{\Theta}[\|\nabla f(\vtheta)\|_2^2]\big)^2 }.
\end{eqnarray}
With $T$ chosen as the smallest integer satisfying \eqref{the necessary condition for the number of iterations main paper}, the total measurement budget of the SPSA estimator is
\begin{eqnarray}
\label{the necessary condition for the number of measurements shots for SPSA main paper}
N_{\textup{SPSA}} = O\bigg(\frac{L M_0\kappa^{\frac{5}{2}}}{ \epsilon_2^{\frac{5}{2}} \big(\E_{\Theta}[\|\nabla f(\vtheta)\|_2^2]\big)^{\frac{5}{2}}} \bigg).
\end{eqnarray}
Here, $\E[\cdot]$ denotes the expectation with respect to the algorithmic randomness generated throughout the SPSA optimization process, including the perturbations $\{\mDelta_t\}_{t=0}^{T-1}$ and the measurement noises $\{\xi(\vtheta_t)\}_{t=0}^{T-1}$. In contrast, $\E_{\Theta}[\cdot]$ denotes the expectation over the parameter distribution $\Theta$ specified above, which serves as a reference scale and is not the distribution of the iterates.
\end{theorem}
The proof is provided in {Appendix}~\ref{Proof of upper bound for gradient}. Theorem~\ref{thm: upper bound of the gradient energy} provides several important insights into the behavior of SPSA-based VQE optimization. First, the result establishes an explicit convergence rate for the expected gradient energy. Specifically, it shows that the gradient energy of at least one iterate generated by SPSA decreases at the rate of $O(T^{-1/2})$ with respect to the number of optimization iterations. In practice, the iterate attaining the minimum in \eqref{expansion of iteration2 another4 upper bound energy2 main paper} is unknown. However, the bound in fact holds for the average $\frac{1}{T}\sum_{t=0}^{T-1}\E[\|\nabla f(\vtheta_t)\|_2^2]$, which upper bounds the minimum. Therefore, if the output $\vtheta_R$ is selected with the index $R$ drawn uniformly at random from $\{0,\ldots,T-1\}$, then $\E[\|\nabla f(\vtheta_R)\|_2^2]\le\kappa T^{-1/2}$, where the expectation is also taken over the random index $R$. Note that this guarantee applies to the randomized selection rule rather than to any fixed iterate. The constant $\kappa$ captures the effects of the Hamiltonian coefficients, the SPSA perturbation scale, the step size, and the measurement budget. This result characterizes how the optimization trajectory approaches a first-order stationary region as the number of iterations increases. We note that Theorem~\ref{thm: upper bound of the gradient energy} guarantees only approximate first-order stationarity. Since small gradients can also occur at saddle points and suboptimal local minima, it does not by itself bound the optimality gap $f(\vtheta)-\inf_{\vtheta'} f(\vtheta')$.

Second, we measure convergence relative to the intrinsic gradient energy scale $\E_{\Theta}[\|\nabla f(\vtheta)\|_2^2]$, which enables a fair comparison across VQE instances with different gradient magnitudes. Since a vanishing gradient is a necessary condition for reaching a critical point, this relative criterion measures relative first-order stationarity. When the intrinsic gradient energy becomes smaller, the number of iterations sufficient to guarantee the same relative reduction increases, which subsequently increases the sufficient total measurement budget. Furthermore, compared with the SNR-based measurement requirement, whose sufficient budget scales with the exponent $3/2$, the convergence analysis introduces an additional factor due to the accumulation of measurement costs over multiple SPSA iterations, resulting in the exponent increasing to $5/2$. These exponents are obtained under the parameter schedules adopted in our analysis, and sharper exponents may be achievable with other schedules or a refined analysis.

Finally, under the BP regime, the intrinsic gradient energy decays exponentially with the number of qubits, namely, $\E_{\Theta}[\|\nabla f(\vtheta)\|_2^2]=O(N\beta^{-n})$. Substituting this scaling into \eqref{the necessary condition for the number of iterations main paper} shows that the number of iterations sufficient to achieve a fixed relative gradient energy level grows exponentially with the system size. Moreover, since the total measurement budget scales polynomially with the inverse of the intrinsic gradient energy, the exponentially decaying gradient energy likewise leads to an exponentially increasing sufficient measurement budget. These results show that the gradient energy of the underlying landscape provides a natural scale governing both the iteration complexity and measurement requirements guaranteed by our analysis. In particular, as the landscape becomes flatter, both the number of iterations and the measurement budget sufficient to reach the same relative stationarity level increase substantially.

\begin{figure*}[ht]
\centering
\subfigure[]{
\begin{minipage}[t]{0.43\textwidth}
\centering
\includegraphics[width=\textwidth]{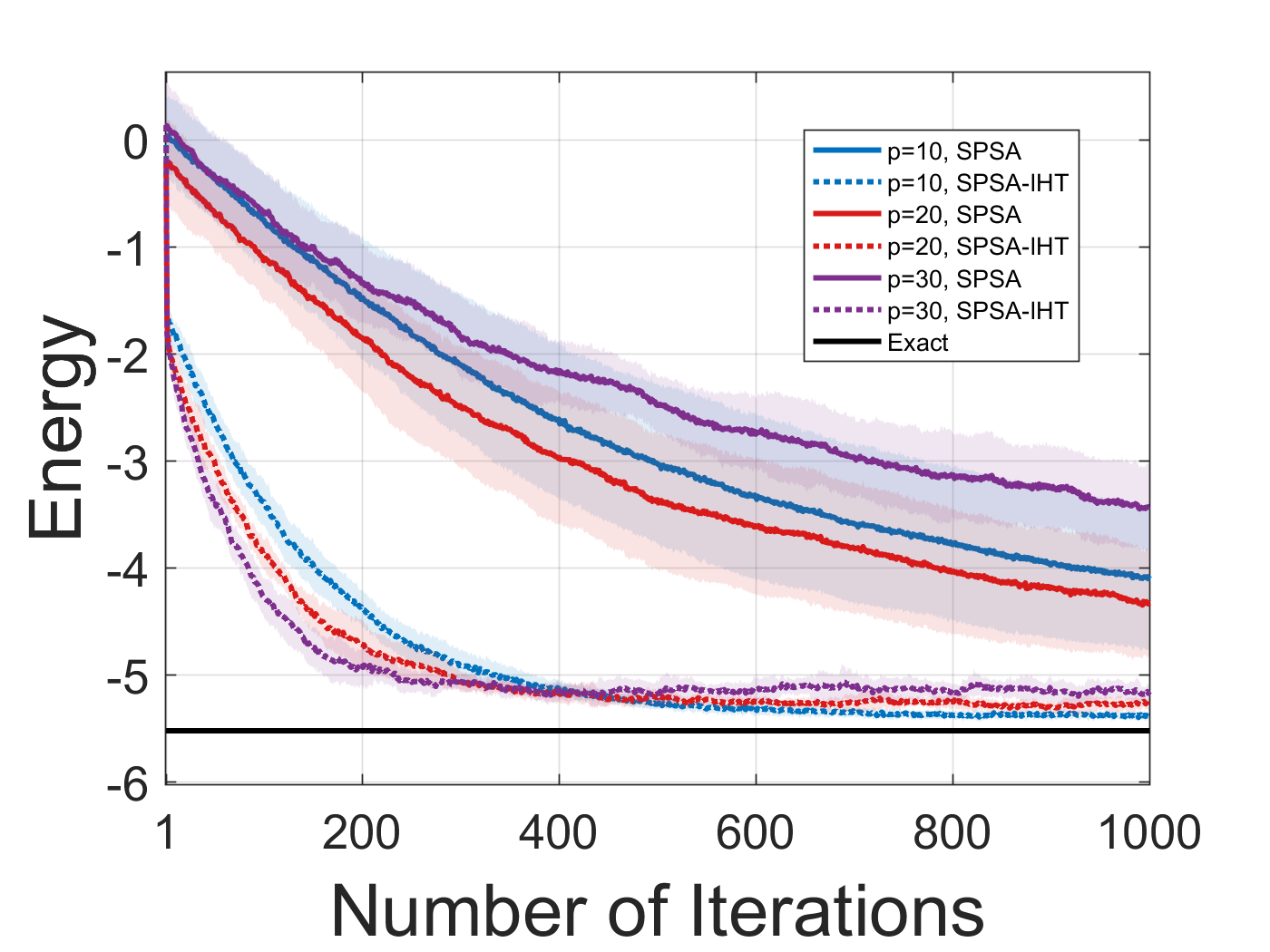}
\end{minipage}
\label{TIM_different_p_IHT}
}
\subfigure[]{
\begin{minipage}[t]{0.43\textwidth}
\centering
\includegraphics[width=\textwidth]{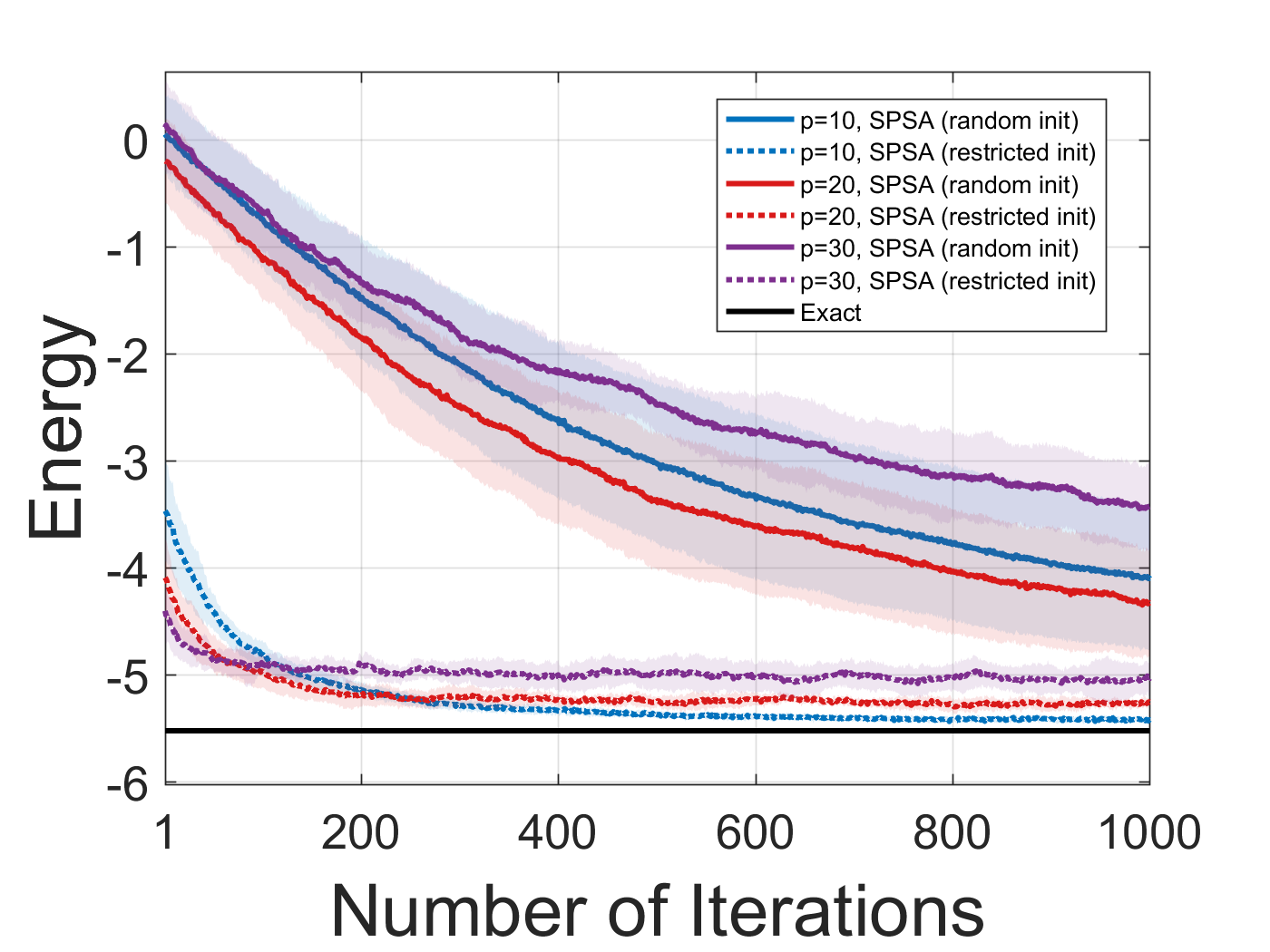}
\end{minipage}
\label{TIM_different_p_vary_init}
}
\subfigure[]{
\begin{minipage}[t]{0.43\textwidth}
\centering
\includegraphics[width=\textwidth]{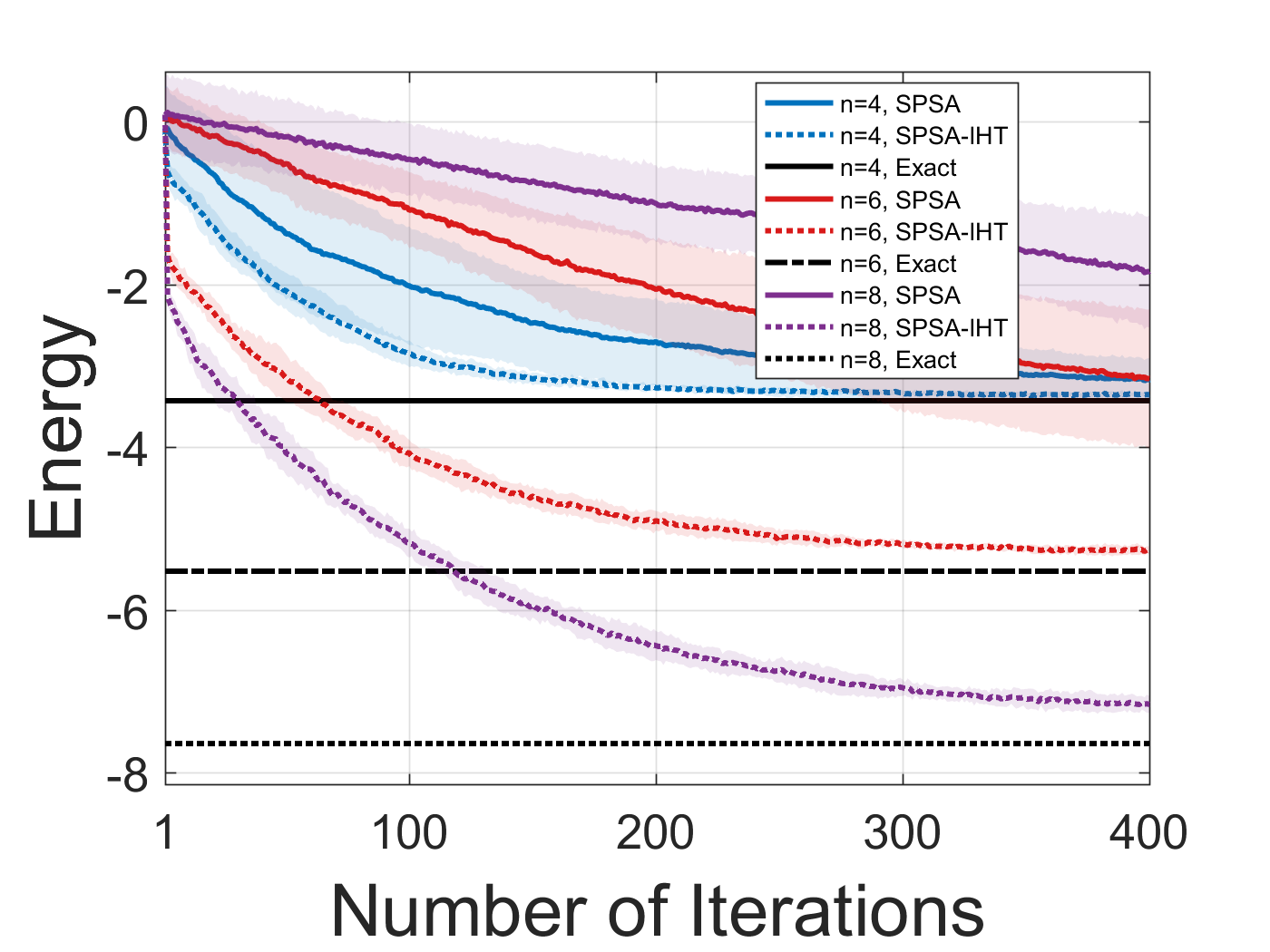}
\end{minipage}
\label{TIM_different_n_IHT}
}
\subfigure[]{
\begin{minipage}[t]{0.43\textwidth}
\centering
\includegraphics[width=\textwidth]{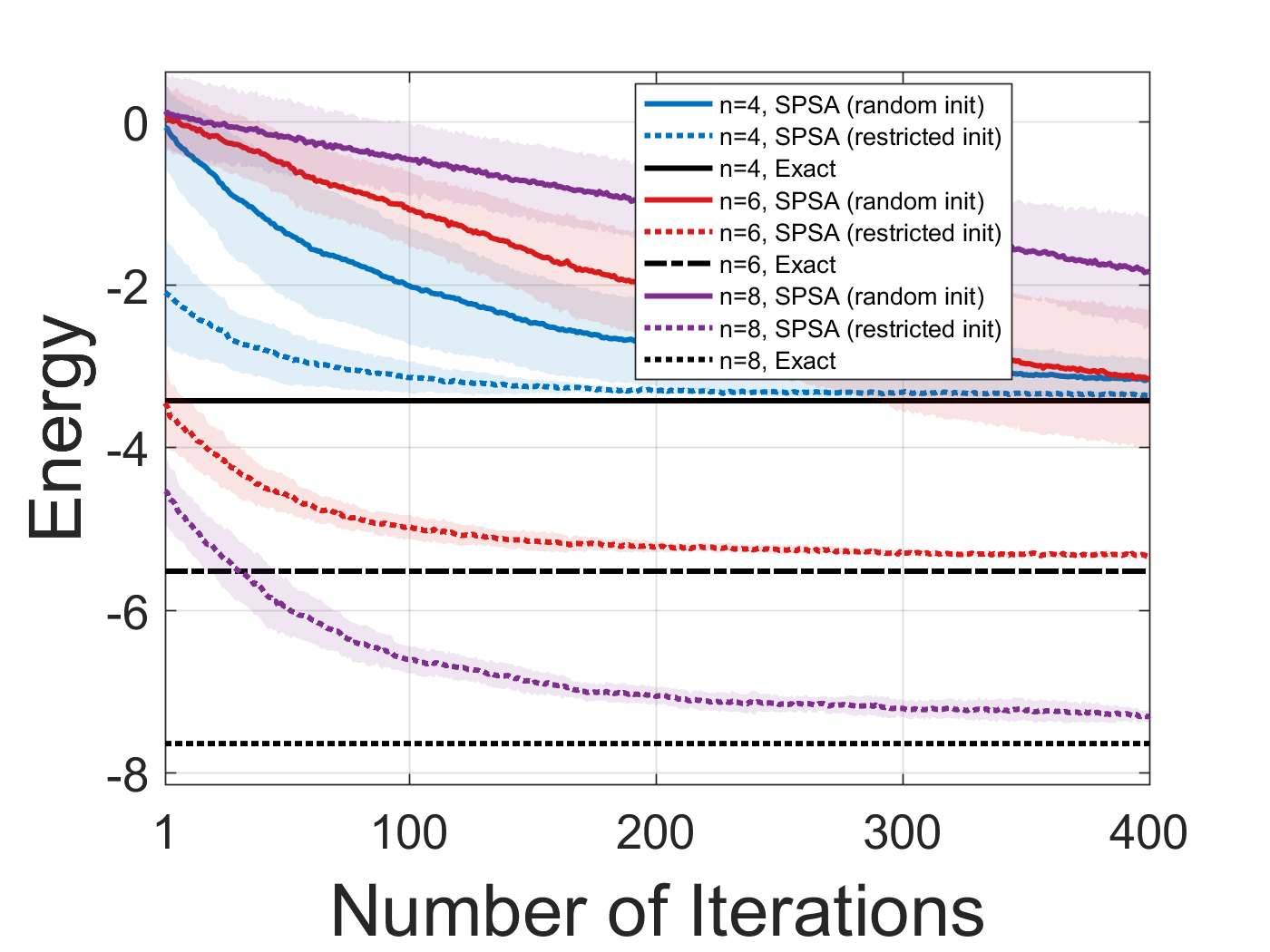}
\end{minipage}
\label{TIM_different_n_vary_init}
}
\caption{Convergence of SPSA for the TIM under the parameter schedules of \Cref{thm: upper bound of the gradient energy}. In (a) and (b), $n=6$ is fixed while $p\in\{10,20,30\}$, corresponding to $N=60,120,180$; in (c) and (d), $p=10$ is fixed while $n\in\{4,6,8\}$, corresponding to $N=40,60,80$. Panels (a) and (c) compare SPSA and SPSA-IHT starting from the same random initial point; panels (b) and (d) compare SPSA, without projection, from the random initialization $U[0,2\pi]^N$ and from the restricted initialization $U[-\pi/p,\pi/p]^N$. Lines show the mean over $10$ independent runs, shaded regions indicate one standard deviation, and black horizontal lines indicate the exact ground-state energies.}
\label{comparison summary}
\end{figure*}

\section{Simulation}
\label{sec: simulation}
In this section, we numerically examine how the intrinsic gradient energy affects SPSA under finite-shot measurements. We consider the transverse-field Ising model (TIM), $\mH_{\mathrm{TIM}}=-J\sum_{i=1}^{n-1}\mZ_i\mZ_{i+1}-h\sum_{i=1}^{n}\mX_i$, where $\mX_i$ and $\mZ_i$ denote the Pauli-$\mX$ and Pauli-$\mZ$ operators acting on the $i$-th qubit, respectively, and we set $J=1$ and $h=0.5$ throughout. We adopt a RealAmplitudes-type ansatz with the linear entanglement scheme \cite{kandala2017hardware}, which consists of $p$ layers, each containing one parametrized rotation gate $\textup{R}_\textup{Y}(\theta)$ on every qubit followed by a linear chain of CNOT gates connecting nearest-neighbor qubits. Hence, on $n$ qubits, the ansatz contains $N=np$ variational parameters. All circuits start from the computational-basis state $|0\rangle^{\otimes n}$. We adopt the fixed-budget parameter schedules of \Cref{thm: upper bound of the gradient energy}, namely $\mu_t=\mu_0T^{-1/2}$, $c_t=c_0T^{-1/8}$, and $M_t=\lceil M_0T^{1/4}\rceil$, which remain constant within each run, with $T=1000$ when $p$ is varied and $T=400$ when $n$ is varied. In addition, we choose the constants empirically as $\mu_0=0.1$, $c_0=0.5$, and $M_0=1000$. For each setting, the reported curves are the mean over $10$ independent runs, each with independently sampled initial parameters, perturbation vectors, and measurement shots, and the shaded regions denote one standard deviation across runs. Because the TIM has fixed couplings, all runs share the same Hamiltonian instance. At each iteration, the energy is evaluated from finite-shot Pauli measurements, and the exact ground-state energy, obtained by exact diagonalization, is shown as a reference.

To obtain a contrasting setting with mitigated BP, we consider an SPSA-iterative hard thresholding (SPSA-IHT) scheme. Reference~\cite[Theorem 1]{wang2024trainability} shows that restricting each parameter $\theta_\ell$ independently and uniformly to $[-a\pi,a\pi]$, with a suitable choice of $a$, can mitigate BP. For the Hamiltonian decomposition $\mH=\sum_{\alpha=1}^{L}w_\alpha\mP_\alpha$, the corresponding intrinsic gradient energy satisfies $\E_{\Theta}[\|\nabla f(\vtheta)\|_2^2]=\Omega\big( \frac{p\, L\, (1+1/S)^{S+1}}{(p+1)^{S+1}} \big)$, where $S$ denotes the maximum number of non-identity Pauli operators among the Pauli strings $\{\mP_\alpha\}_{\alpha=1}^{L}$. Following this result, SPSA-IHT projects each updated parameter onto $[-\pi/p,\pi/p]$ after every iteration, and it starts from the same random initial point $\vtheta_0\sim U[0,2\pi]^N$ as standard SPSA. Since the projection changes both the feasible set and the optimization dynamics, it is not covered by \Cref{thm: upper bound of the gradient energy}. Furthermore, to separate the effect of the initialization from that of the projection, we further run standard SPSA, without projection, from two different initializations: a random initialization $\vtheta_0\sim U[0,2\pi]^N$ and a restricted initialization $\vtheta_0\sim U[-\pi/p,\pi/p]^N$. To ensure that the two initializations are compared under identical conditions, in each run we draw $\vu\sim U[0,1]^N$ and set $\vtheta_0=2\pi\vu$ and $\vtheta_0=(\pi/p)(2\vu-1)$, respectively, and use the same perturbation vectors $\mDelta_t$ for both initializations.

As shown in \Cref{TIM_different_p_IHT}, increasing the number of layers $p$ leads to slower convergence of SPSA, which becomes particularly pronounced for $p=30$, whereas SPSA-IHT converges rapidly to energies close to the ground-state energy for all values of $p$. A similar trend is observed in \Cref{TIM_different_n_IHT}: as the number of qubits $n$ increases, the convergence of SPSA slows down substantially, while the degradation is much less pronounced for SPSA-IHT. Notably, the energy of SPSA-IHT drops abruptly after the first iteration. This is because most components of the random initial point lie outside $[-\pi/p,\pi/p]$ and are mapped onto the boundary of this interval by the first projection, so that the parameters are moved into the restricted region at the very beginning of the optimization. To further examine the role of the initialization, \Cref{TIM_different_p_vary_init,TIM_different_n_vary_init} compare SPSA from the random and restricted initializations. Standard SPSA initialized in the restricted region, without any subsequent projection, exhibits nearly the same convergence behavior as SPSA-IHT for all values of $p$ and $n$. This suggests that the improved convergence stems mainly from placing the parameters in the restricted region rather than from the subsequent projections. According to \cite[Theorem 1]{wang2024trainability}, restricting the parameters to $[-\pi/p,\pi/p]$ yields a larger intrinsic gradient energy than the unrestricted setting. Hence, the observed behavior is consistent with our theoretical results, which indicate that a larger gradient energy reduces the number of iterations and measurements sufficient for convergence. Overall, the numerical results support that the intrinsic gradient energy of the optimization landscape plays an important role in the convergence of SPSA under finite-shot measurements.

\section{Conclusion}
\label{sec: Conclusion}

In this paper, we establish a theoretical framework for understanding how the BP affects SPSA optimization in VQEs under finite-shot measurements. By characterizing the statistical properties of the SPSA gradient estimator, relating its mean squared error to the gradient signal through an SNR analysis, and analyzing the resulting convergence behavior, we establish how the landscape-induced decay of gradient energy translates into both the iteration complexity and the measurement budget sufficient to achieve a prescribed level of relative first-order stationarity. In particular, we show that exponentially vanishing gradient energy leads to exponential growth in both the sufficient number of iterations and the sufficient total measurement budget. These findings provide a quantitative understanding of how the structure of BP influences the optimization complexity of SPSA, linking the geometry of the optimization landscape to the computational resources required for practical variational quantum optimization. Several questions remain open. First, our results provide sufficient guarantees under specific parameter schedules; establishing matching lower bounds would clarify whether the resulting complexity bounds are tight. Second, first-order stationarity alone does not guarantee a small optimality gap. Establishing conditions under which the derived stationarity guarantees translate into energy accuracy of the final VQE output is therefore an important direction for future work. 

\section{Acknowledgment}

ZQ gratefully acknowledges support from the MICDE Research Scholars Program at the University of Michigan.

\bibliography{reference1}

\appendix
\onecolumngrid

\section{Proof of Proposition \ref{Proposition: Hessian Property}}
\label{Proof of Hessian property}

Before presenting the proof, we first introduce several preliminary definitions and identities that will be used throughout the derivation.
For each Pauli term $\mP_\alpha$, define $f_\alpha(\vtheta):=\langle\vphi_0|\mU^\dagger(\vtheta)\mP_\alpha \mU(\vtheta)|\vphi_0\rangle$ with $\mU(\vtheta) = \mU_N(\theta_N)\cdots\mU_1(\theta_1)$. Then the objective function admits the decomposition $ f(\vtheta) = \sum_{\alpha=1}^{L}w_\alpha f_\alpha(\vtheta) $.
For convenience, we further introduce the effective generator
\begin{eqnarray}
\label{dressed generator}
\wt \mG_\ell(\vtheta):=i\,(\partial_\ell \mU(\vtheta))\,\mU^\dagger(\vtheta),
\end{eqnarray}
where $\partial_\ell \mU(\vtheta) =
\partial\mU(\vtheta)/\partial\theta_\ell$ denotes the partial derivative of $\mU(\vtheta)$ with respect to $\theta_\ell$. Right-multiplying \eqref{dressed generator} by
$\mU(\vtheta)$ gives
\begin{eqnarray}
\label{equation of partial U}
&\!\!\!\!\!\!\!\!&\wt \mG_\ell(\vtheta) \mU(\vtheta) = i\,(\partial_\ell \mU(\vtheta))\,\mU^\dagger(\vtheta) \mU(\vtheta) = i\,(\partial_\ell \mU(\vtheta))\nonumber\\
&\!\!\!\!\Longrightarrow\!\!\!\!& \partial_\ell \mU(\vtheta) = - i \wt \mG_\ell(\vtheta) \mU(\vtheta).
\end{eqnarray}
Differentiating the identity $\mU(\vtheta)\mU^\dagger(\vtheta)=\mId$ with respect to $\theta_\ell$ yields
\begin{eqnarray}
\label{partial derivation of UU}
(\partial_\ell \mU(\vtheta))\mU^\dagger(\vtheta) + \mU(\vtheta)(\partial_\ell \mU^\dagger(\vtheta)) = {\bf 0}.
\end{eqnarray}
Substituting \eqref{equation of partial U} into \eqref{partial derivation of UU}, we obtain
\begin{eqnarray}
\label{equation of partial U dagger}
&\!\!\!\!\!\!\!\!&(-i\wt \mG_\ell(\vtheta) \mU(\vtheta))\mU^\dagger(\vtheta)  + \mU(\vtheta)(\partial_\ell \mU^\dagger(\vtheta))={\bf 0}\nonumber\\
&\!\!\!\!\Longrightarrow\!\!\!\!& \mU(\vtheta)(\partial_\ell \mU^\dagger(\vtheta)) =  i\wt \mG_\ell(\vtheta) \nonumber\\
&\!\!\!\!\Longrightarrow\!\!\!\!& \mU^\dagger(\vtheta)\mU(\vtheta)(\partial_\ell \mU^\dagger(\vtheta)) =  i\mU^\dagger(\vtheta)\wt \mG_\ell(\vtheta) \nonumber\\
&\!\!\!\!\Longrightarrow\!\!\!\!& \partial_\ell \mU^\dagger(\vtheta) =  i\mU^\dagger(\vtheta)\wt \mG_\ell(\vtheta).
\end{eqnarray}
We are now ready to prove Proposition~\ref{Proposition: Hessian Property}.

\begin{proof}
\medskip
\noindent\textbf{Step 1: First-order derivative formula.}

Differentiating $f_\alpha(\vtheta)$ with respect to $\theta_\ell$ yields
\begin{eqnarray}
\label{first derivation of f}
\partial_\ell f_\alpha(\vtheta) = \langle\vphi_0|(\partial_\ell\mU^\dagger(\vtheta)) \mP_\alpha \mU(\vtheta)|\vphi_0\rangle + \langle\vphi_0|\mU^\dagger(\vtheta)\mP_\alpha \partial_\ell\mU(\vtheta)|\vphi_0\rangle.
\end{eqnarray}
Substituting \eqref{equation of partial U} and  \eqref{equation of partial U dagger} into \eqref{first derivation of f}, we obtain
\begin{eqnarray}
\label{first derivation of f1}
\partial_\ell f_\alpha(\vtheta) &\!\!\!\!=\!\!\!\!& i\langle\vphi_0|\mU^\dagger(\vtheta)\wt \mG_\ell(\vtheta) \mP_\alpha \mU(\vtheta)|\vphi_0\rangle - i \langle\vphi_0|\mU^\dagger(\vtheta)\mP_\alpha \wt \mG_\ell(\vtheta) \mU(\vtheta)|\vphi_0\rangle\nonumber\\
&\!\!\!\!=\!\!\!\!& i\langle\vphi_0|\mU^\dagger(\vtheta) [\wt \mG_\ell(\vtheta), \mP_\alpha] \mU(\vtheta)|\vphi_0\rangle,
\end{eqnarray}
where $[\wt \mG_\ell(\vtheta), \mP_\alpha] := \wt \mG_\ell(\vtheta) \mP_\alpha - \mP_\alpha \wt \mG_\ell(\vtheta)$ denotes the commutator.

\medskip
\noindent\textbf{Step 2: Explicit expression for $\wt{\mG}_\ell$.}

To derive an explicit expression for the effective generator, we decompose $\mU(\vtheta) = \mR_{\ell}(\vtheta) \mU_{\ell}(\theta_{\ell}) \mL_{\ell}(\vtheta)$ where $\mR_{\ell}(\vtheta) := \mU_N(\theta_N)\cdots \mU_{\ell+1}(\theta_{\ell+1})$ and $\mL_{\ell}(\vtheta):= \mU_{\ell-1}(\theta_{\ell-1})\cdots \mU_1(\theta_1)$. Using $\mU_\ell(\theta_\ell)=e^{-i\theta_\ell \mG_\ell}$, we obtain
\begin{eqnarray}
\label{derivation of Ul in U}
\partial_{\ell} \mU(\vtheta) = \mR_{\ell}(\vtheta) \partial_{\ell}\mU_{\ell}(\theta_{\ell}) \mL_{\ell}(\vtheta) = -i \mR_{\ell}(\vtheta) \mG_\ell \mU_\ell(\theta_\ell) \mL_{\ell}(\vtheta) = -i \mR_{\ell}(\vtheta) \mG_\ell \mR_{\ell}^\dagger(\vtheta) \mU(\vtheta).
\end{eqnarray}
Comparing the above identity with \eqref{equation of partial U}, we conclude that
\begin{eqnarray}
\label{equation of wt G}
\wt \mG_\ell(\vtheta) = \mR_{\ell}(\vtheta) \mG_\ell \mR_{\ell}^\dagger(\vtheta).
\end{eqnarray}
Since $\mR_\ell(\vtheta)$ is unitary, it follows that
\begin{eqnarray}
\label{spectral norm of wt G}
\|\wt \mG_\ell(\vtheta)\|\leq \|\mR_{\ell}(\vtheta)\| \|\mG_\ell\| \|\mR_{\ell}^\dagger(\vtheta)\|\le 1.
\end{eqnarray}
Moreover, since $\wt \mG_\ell(\vtheta)$ depends only on $\theta_{\ell+1},\ldots,\theta_N$, we have
\begin{eqnarray}
\label{derivation of wt G 1}
\partial_a \wt \mG_\ell(\vtheta) \equiv 0,\qquad a\le\ell.
\end{eqnarray}
For $a>\ell$, differentiating $\mR_\ell(\vtheta)$ yields
\begin{eqnarray}
\label{derivation of wt G 2}
\partial_a \mR_{\ell}(\vtheta) &\!\!\!\!=\!\!\!\!& -i \mU_N(\theta_N)  \cdots  \mU_{a+1}(\theta_{a+1})  \mG_a \mU_a(\theta_a)   \cdots  \mU_{\ell+1}(\theta_{\ell+1})\nonumber\\
&\!\!\!\!=\!\!\!\!& - i \wt \mG_a(\vtheta)\mR_{\ell}(\vtheta).
\end{eqnarray}
Taking the Hermitian transpose gives
\begin{eqnarray}
\label{derivation of wt G 3}
\partial_a \mR_{\ell}^\dagger(\vtheta) = (\partial_a \mR_{\ell}(\vtheta))^\dagger = i \mR_{\ell}^\dagger(\vtheta) \wt \mG_a(\vtheta).
\end{eqnarray}
Substituting \eqref{derivation of wt G 2} and \eqref{derivation of wt G 3} into the derivative of \eqref{equation of wt G} gives
\begin{eqnarray}
\label{derivation of wt G 4}
\partial_a \wt \mG_\ell(\vtheta) &\!\!\!\!=\!\!\!\!& \partial_a \mR_{\ell}(\vtheta) \mG_\ell \mR_{\ell}^\dagger(\vtheta) + \mR_{\ell}(\vtheta) \mG_\ell \partial_a \mR_{\ell}^\dagger(\vtheta)\nonumber\\
&\!\!\!\!=\!\!\!\!& -i \wt \mG_a(\vtheta)\mR_{\ell}(\vtheta) \mG_\ell \mR_{\ell}^\dagger(\vtheta) + i \mR_{\ell}(\vtheta) \mG_\ell \mR_{\ell}^\dagger(\vtheta) \wt \mG_a(\vtheta)\nonumber\\
&\!\!\!\!=\!\!\!\!& i \wt \mG_\ell(\vtheta) \wt \mG_a(\vtheta) - i \wt \mG_a(\vtheta) \wt \mG_\ell(\vtheta)\nonumber\\
&\!\!\!\!=\!\!\!\!& i [ \wt \mG_\ell(\vtheta),  \wt \mG_a(\vtheta) ].
\end{eqnarray}
Consequently, we can derive
\begin{eqnarray}
\label{derivation of wt G bound}
\|\partial_a \wt \mG_\ell(\vtheta)\| = \|[ \wt \mG_\ell(\vtheta),  \wt \mG_a(\vtheta) ]\|\leq 2.
\end{eqnarray}

\medskip
\noindent\textbf{Step 3: Second-order derivative formula.}

We first consider $\ell\geq j$. Differentiating \eqref{first derivation of f1} with respect to $\theta_j$ yields
\begin{eqnarray}
\label{second derivation of f1}
\partial_{\ell j} f_\alpha(\vtheta) &\!\!\!\!=\!\!\!\!&    i\langle\vphi_0|\partial_{j}\mU^\dagger(\vtheta) [\wt \mG_\ell(\vtheta), \mP_\alpha] \mU(\vtheta)|\vphi_0\rangle + i\langle\vphi_0|\mU^\dagger(\vtheta) [\partial_{j}\wt \mG_\ell(\vtheta), \mP_\alpha] \mU(\vtheta)|\vphi_0\rangle\nonumber\\
&\!\!\!\!\!\!\!\!&  + i\langle\vphi_0|\mU^\dagger(\vtheta) [\wt \mG_\ell(\vtheta), \mP_\alpha] \partial_{j}\mU(\vtheta)|\vphi_0\rangle \nonumber\\
&\!\!\!\!=\!\!\!\!& -\langle\vphi_0|\mU^\dagger(\vtheta)\wt \mG_j(\vtheta) [\wt \mG_\ell(\vtheta), \mP_\alpha] \mU(\vtheta)|\vphi_0\rangle + i\langle\vphi_0|\mU^\dagger(\vtheta) [\partial_{j}\wt \mG_\ell(\vtheta), \mP_\alpha] \mU(\vtheta)|\vphi_0\rangle\nonumber\\
&\!\!\!\!\!\!\!\!&  +  \langle\vphi_0|\mU^\dagger(\vtheta) [\wt \mG_\ell(\vtheta), \mP_\alpha] \wt \mG_j(\vtheta) \mU(\vtheta) |\vphi_0\rangle\nonumber\\
&\!\!\!\!=\!\!\!\!& -\langle\vphi_0|\mU^\dagger(\vtheta) [\wt \mG_j(\vtheta),  [\wt \mG_\ell(\vtheta), \mP_\alpha] ] \mU(\vtheta)|\vphi_0\rangle + i\langle\vphi_0|\mU^\dagger(\vtheta) [\partial_{j}\wt \mG_\ell(\vtheta), \mP_\alpha] \mU(\vtheta)|\vphi_0\rangle\nonumber\\
&\!\!\!\!=\!\!\!\!& -\langle\vphi_0|\mU^\dagger(\vtheta) [\wt \mG_j(\vtheta),  [\wt \mG_\ell(\vtheta), \mP_\alpha] ] \mU(\vtheta)|\vphi_0\rangle, \quad \ell\geq j,
\end{eqnarray}
where the second equality follows from \eqref{equation of partial U} and \eqref{equation of partial U dagger}, and the last equality follows from \eqref{derivation of wt G 1}. 

Note that the last equality relies on $\partial_j\wt\mG_\ell(\vtheta)=0$, which holds only for $j\le\ell$ by \eqref{derivation of wt G 1}; for $j>\ell$, this term does not vanish in general, and the order of $\wt\mG_j(\vtheta)$ and $\wt\mG_\ell(\vtheta)$ in the nested commutator cannot be exchanged. For arbitrary indices $\ell,j\in\{1,\ldots,N\}$, we therefore denote by $\ell'=\max\{\ell,j\}$ and $j'=\min\{\ell,j\}$ the sorted indices. Since $\partial_{\ell j}f_\alpha(\vtheta)=\partial_{\ell' j'}f_\alpha(\vtheta)$ by Clairaut's theorem, applying \eqref{second derivation of f1} to $(\ell',j')$ gives
\begin{eqnarray}
\label{second derivation of f1 sorted}
\partial_{\ell j} f_\alpha(\vtheta) = -\langle\vphi_0|\mU^\dagger(\vtheta) [\wt \mG_{j'}(\vtheta),  [\wt \mG_{\ell'}(\vtheta), \mP_\alpha] ] \mU(\vtheta)|\vphi_0\rangle, \quad \ell,j\in\{1,\ldots,N\},
\end{eqnarray}
where the generator with the smaller index always appears in the outer commutator.

\medskip
\noindent\textbf{Step 4: Third-order derivative formula.}

Considering $\ell \geq j \geq k$, we differentiate \eqref{second derivation of f1} with respect to $\theta_k$ and have
\begin{eqnarray}
\label{third derivation of f1}
&\!\!\!\!\!\!\!\!& \partial_{\ell j k} f_\alpha(\vtheta)\nonumber\\
&\!\!\!\!=\!\!\!\!& -\langle\vphi_0|\partial_{k}\mU^\dagger(\vtheta) [\wt \mG_j(\vtheta),  [\wt \mG_\ell(\vtheta), \mP_\alpha] ] \mU(\vtheta)|\vphi_0\rangle -\langle\vphi_0|\mU^\dagger(\vtheta) [\partial_{k}\wt \mG_j(\vtheta),  [\wt \mG_\ell(\vtheta), \mP_\alpha] ] \mU(\vtheta)|\vphi_0\rangle  \nonumber\\
&\!\!\!\!\!\!\!\!& -\langle\vphi_0|\mU^\dagger(\vtheta) [\wt \mG_j(\vtheta),  [\partial_{k}\wt \mG_\ell(\vtheta), \mP_\alpha] ] \mU(\vtheta)|\vphi_0\rangle -\langle\vphi_0|\mU^\dagger(\vtheta) [\wt \mG_j(\vtheta),  [\wt \mG_\ell(\vtheta), \mP_\alpha] ] \partial_{k}\mU(\vtheta)|\vphi_0\rangle\nonumber\\
&\!\!\!\!=\!\!\!\!& -i \langle\vphi_0| \mU^\dagger(\vtheta) \wt \mG_k(\vtheta) [\wt \mG_j(\vtheta),  [\wt \mG_\ell(\vtheta), \mP_\alpha] ] \mU(\vtheta)|\vphi_0\rangle + i \langle\vphi_0|\mU^\dagger(\vtheta) [\wt \mG_j(\vtheta),  [\wt \mG_\ell(\vtheta), \mP_\alpha] ] \wt \mG_k(\vtheta) \mU(\vtheta)|\vphi_0\rangle\nonumber\\
&\!\!\!\!\!\!\!\!& -\langle\vphi_0|\mU^\dagger(\vtheta) [\partial_{k}\wt \mG_j(\vtheta),  [\wt \mG_\ell(\vtheta), \mP_\alpha] ] \mU(\vtheta)|\vphi_0\rangle -\langle\vphi_0|\mU^\dagger(\vtheta) [\wt \mG_j(\vtheta),  [\partial_{k}\wt \mG_\ell(\vtheta), \mP_\alpha] ] \mU(\vtheta)|\vphi_0\rangle\nonumber\\
&\!\!\!\!=\!\!\!\!& -i \langle\vphi_0| \mU^\dagger(\vtheta) [\wt \mG_k(\vtheta),  [\wt \mG_j(\vtheta),  [\wt \mG_\ell(\vtheta), \mP_\alpha] ] ] \mU(\vtheta)|\vphi_0\rangle, \quad \ell\geq j \geq k,
\end{eqnarray}
where the second equality follows from \eqref{equation of partial U} and \eqref{equation of partial U dagger}, and the last equality follows from \eqref{derivation of wt G 1}.

Here, the last equality uses $\partial_k\wt\mG_j(\vtheta)=\partial_k\wt\mG_\ell(\vtheta)=0$, which holds for $k\le j\le\ell$ by \eqref{derivation of wt G 1}. For arbitrary indices $\ell,j,k\in\{1,\ldots,N\}$, let $\ell'\ge j'\ge k'$ denote the same indices sorted in non-increasing order. Since $\partial_{\ell jk}f_\alpha(\vtheta)=\partial_{\ell' j' k'}f_\alpha(\vtheta)$ by Clairaut's theorem, applying \eqref{third derivation of f1} to $(\ell',j',k')$ gives
\begin{eqnarray}
\label{third derivation of f1 sorted}
\partial_{\ell j k} f_\alpha(\vtheta) = -i \langle\vphi_0| \mU^\dagger(\vtheta) [\wt \mG_{k'}(\vtheta),  [\wt \mG_{j'}(\vtheta),  [\wt \mG_{\ell'}(\vtheta), \mP_\alpha] ] ] \mU(\vtheta)|\vphi_0\rangle, \quad \ell,j,k\in\{1,\ldots,N\}.
\end{eqnarray}
Although the ordering of the generators in \eqref{third derivation of f1 sorted} is essential, the norm bound below does not depend on it. Furthermore, we have
\begin{eqnarray}
\label{spectral norm third derivation of f2}
|\partial_{\ell j k} f(\vtheta)| &\!\!\!\!\leq\!\!\!\!& \sum_{\alpha=1}^{L} |w_\alpha| |\partial_{\ell j k}f_\alpha(\vtheta) |\nonumber\\
&\!\!\!\!\leq\!\!\!\!& \sum_{\alpha=1}^{L} |w_\alpha| \| [\wt \mG_{k'}(\vtheta),  [\wt \mG_{j'}(\vtheta),  [\wt \mG_{\ell'}(\vtheta), \mP_\alpha] ] ]  \|\nonumber\\
&\!\!\!\!\leq\!\!\!\!& \sum_{\alpha=1}^{L} 8 |w_\alpha|\|\wt \mG_{\ell'}(\vtheta) \|\|\wt \mG_{j'}(\vtheta) \|\|\wt \mG_{k'}(\vtheta) \|  \|\mP_\alpha\|\nonumber\\
&\!\!\!\!\leq\!\!\!\!& 8 \|\vw\|_1.
\end{eqnarray}

\medskip
\noindent\textbf{Step 5: Hessian Lipschitz constant.}

Fix $\vx,\vy\in\R^{N}$ and define $\mDelta:=\vx-\vy$. For each $(\ell,j)$, the function $t\mapsto\partial_{\ell j}f(\vy+t\mDelta)$ is differentiable on $[0,1]$. Hence, by the mean value theorem, there exists $t_{\ell j}\in(0,1)$ such that
\begin{eqnarray}
\label{mean value theorem second order}
\partial_{\ell j}f(\vx) - \partial_{\ell j}f(\vy) =   \sum_{k=1}^N \partial_{\ell jk}f(\vy+t_{\ell j}\mDelta)\,\Delta_k .
\end{eqnarray}
Applying \eqref{spectral norm third derivation of f2} and the triangle inequality gives
\begin{eqnarray}
\label{Heesian L constant of second order}
\big|\big(\nabla^2f(\vx)-\nabla^2f(\vy)\big)_{\ell j}\big| &\!\!\!\!=\!\!\!\!& \big|\partial_{\ell j}f(\vx)-\partial_{\ell j}f(\vy)\big|\nonumber\\
&\!\!\!\!\le\!\!\!\!& \sum_{k=1}^N 8\|\vw\|_1\,|\Delta_k| = 8\|\vw\|_1\,\|\mDelta\|_1 \leq 8\sqrt{N}\|\vw\|_1\,\|\mDelta\|_2.
\end{eqnarray}
Since $\nabla^2f(\vx)-\nabla^2f(\vy)$ is symmetric, its spectral norm is bounded by the maximum absolute row sum. Therefore, we have
\begin{eqnarray}
\label{Heesian L constant of second order1}
\|\nabla^2f(\vx)-\nabla^2f(\vy)\| \leq \max_{1\leq \ell \leq N}\sum_{j=1}^{N}\big|\big(\nabla^2f(\vx)-\nabla^2f(\vy)\big)_{\ell j}\big| \leq   8N^{\frac{3}{2}}\|\vw\|_1\,\|\vx - \vy\|_2.
\end{eqnarray}
\end{proof}

\section{Proof of \Cref{lem:noise-model}}
\label{proof of lemma different noise properties}

\begin{proof}

\medskip
\noindent\textbf{Proof of \eqref{properties of noise1}.} The expectation of a single-shot measurement outcome follows directly from the
Born rule. Specifically, we have
\begin{eqnarray}
\label{expectation of Pauli observables pm}
\E[X_{\alpha,k}^s \mid \mDelta] &\!\!=\!\!& (+1)\cdot \trace(\vrho(\vtheta^s){\bm \Pi}_{\alpha}^+) + (-1)\cdot \trace(\vrho(\vtheta^s){\bm \Pi}_{\alpha}^-)\nonumber\\
&\!\!=\!\!&  \trace\big(\vrho(\vtheta^s)\, \mP_\alpha\big).
\end{eqnarray}
The noiseless objective value satisfies
\begin{eqnarray}
\label{expectation of real output}
f(\vtheta^s) = \trace(\vrho(\vtheta^s) \mH) = \sum_{\alpha=1}^L w_\alpha\, \trace(\vrho(\vtheta^s) \mP_\alpha).
\end{eqnarray}
Combining the above relations with the definition of the empirical estimator $\hat f^s := \sum_{\alpha=1}^L w_\alpha \hat P_\alpha^s$, we obtain
\begin{eqnarray}
\label{expectation of noisy output}
\E[\hat f^s \mid \mDelta] &\!\!=\!\!& \sum_{\alpha=1}^L w_\alpha\, \E[\hat P_\alpha^s \mid \mDelta]\nonumber\\
&\!\!=\!\!& \sum_{\alpha=1}^L w_\alpha\, \frac{1}{M}\sum_{k=1}^M \E[X_{\alpha,k}^s\mid\mDelta]\nonumber\\
&\!\!=\!\!& \sum_{\alpha=1}^L w_\alpha\, \trace\big(\vrho(\vtheta^s)\, \mP_\alpha\big)\nonumber\\
&\!\!=\!\!&f(\vtheta^s),
\end{eqnarray}
where the third equality follows from \eqref{expectation of Pauli observables pm}, and the last equality follows from \eqref{expectation of real output}.

It follows that the oracle noise term $\xi(\vtheta^s) := \hat f^s-f(\vtheta^s)$ satisfies
\begin{eqnarray}
\label{expectation of noise term appendix}
\E[\xi(\vtheta^s) \mid \mDelta] = \E[\hat f^s \mid \mDelta] - f(\vtheta^s)  = 0.
\end{eqnarray}

\medskip
\noindent\textbf{Proof of \eqref{properties of noise2}.} By the definition of the empirical estimator $\hat f^s := \sum_{\alpha=1}^L w_\alpha \hat P_\alpha^s$, the conditional variance of $\xi(\vtheta^s)$ can be expanded as
\begin{eqnarray}
\label{variance of noisy output}
\mathrm{Var}(\xi(\vtheta^s) \mid \mDelta) &\!\!=\!\!& \mathrm{Var}(\hat f^s - f(\vtheta^s) \mid \mDelta)
= \mathrm{Var}(\hat f^s \mid \mDelta)\nonumber\\
&\!\!=\!\!& \sum_{\alpha=1}^L w_\alpha^2\, \mathrm{Var}(\hat P_\alpha^s\mid\mDelta)
+ \sum_{\alpha\ne\beta} w_\alpha w_\beta\,
\mathrm{Cov}(\hat P_\alpha^s, \hat P_\beta^s \mid \mDelta)\nonumber\\
&\!\!=\!\!& \sum_{\alpha=1}^L w_\alpha^2\, \mathrm{Var}(\hat P_\alpha^s\mid\mDelta) = \sum_{\alpha=1}^L \frac{w_\alpha^2}{M^2}\,\mathrm{Var}\Big(\sum_{k=1}^M X_{\alpha,k}^s \,\Big|\,\mDelta\Big)\nonumber\\
&\!\!=\!\!& \sum_{\alpha=1}^L \frac{w_\alpha^2}{M^2}\, \bigg( \sum_{k=1}^M \mathrm{Var}(X_{\alpha,k}^s\mid\mDelta)
+ \sum_{k\ne j} \mathrm{Cov}(X_{\alpha,k}^s, X_{\alpha,j}^s \mid \mDelta)  \bigg)\nonumber\\
&\!\!=\!\!& \sum_{\alpha=1}^L \frac{w_\alpha^2}{M^2}\,   \sum_{k=1}^M \mathrm{Var}(X_{\alpha,k}^s\mid\mDelta)\nonumber\\
&\!\!\leq\!\!& \frac{\|\vw\|_2^2}{M},
\end{eqnarray}
where the third equality follows from the conditional independence of
$\hat P_\alpha^s$ and $\hat P_\beta^s$ for $\alpha\neq\beta$, and the fifth
equality follows from the conditional independence of
$X_{\alpha,k}^s$ and $X_{\alpha,j}^s$ for $k\neq j$. Both facts are direct
consequences of the measurement protocol, under which all measurement outcomes
are mutually independent conditional on $\mDelta$ across different Pauli
observables, perturbation signs, and shot indices. The final inequality follows
from $\mathrm{Var}(X_{\alpha,k}^s \mid \mDelta)
= \E[(X_{\alpha,k}^s)^2\mid\mDelta] - \big(\E[X_{\alpha,k}^s\mid\mDelta]\big)^2
= 1 - \big(\E[X_{\alpha,k}^s\mid\mDelta]\big)^2 \le 1$.

\medskip
\noindent\textbf{Proof of  \eqref{independence of noise}.} By construction, $\xi(\vtheta^{+})$ is determined solely by the measurement outcomes $\{X_{\alpha,k}^{+}\}_{\alpha,k}$, whereas $\xi(\vtheta^{-})$ is determined solely by the outcomes $\{X_{\alpha,k}^{-}\}_{\alpha,k}$. Since the measurement outcomes associated with different perturbation signs are mutually independent conditional on $\mDelta$, the two collections of random variables are independent given $\mDelta$. Therefore, any functions of these two independent collections are also independent conditional on $\mDelta$, which directly yields
\begin{eqnarray}
\label{independence of noise appendix}
\xi(\vtheta^+) \perp \xi(\vtheta^-) \mid \mDelta.
\end{eqnarray}
\end{proof}

\section{Proof of \Cref{thm:main property of noisy gradient}}
\label{Proof of theorem main property of noisy gradient}
\begin{proof}

We begin by decomposing the SPSA gradient estimator into a deterministic finite-difference term and a measurement noise term. By $\hat f(\vtheta) = f(\vtheta)+\xi(\vtheta)$ and $ \hat g_{\ell}(\vtheta_t)  = \frac{\hat f(\vtheta_t + c_t\mDelta_t) - \hat f(\vtheta_t- c_t\mDelta_t) }{2c_t\Delta_{\ell,t}}$, we have
\begin{eqnarray}
\label{decomposition of SPSA estimator1}
\hat g_\ell(\vtheta)  &\!\!=\!\!&  \frac{f(\vtheta + c\mDelta)-f(\vtheta - c\mDelta)}{2c\Delta_\ell} + \frac{\xi(\vtheta^+) - \xi(\vtheta^-)}{2c\Delta_\ell}\nonumber\\
&\!\!:=\!\!& \frac{f(\vtheta + c\mDelta)-f(\vtheta - c\mDelta)}{2c\Delta_\ell} + \delta_{\ell}(\vtheta).
\end{eqnarray}
The finite-difference term in \eqref{decomposition of SPSA estimator1} can be analyzed using Taylor's theorem with Lagrange remainder. There exist $\zeta^+,\zeta^-\in(0,1)$, depending only on $\mDelta$, such that
\begin{eqnarray}
\label{Taylor expansion of noiseless SPSA estimator1}
f(\vtheta \pm c\mDelta) = f(\vtheta) \pm c \nabla f(\vtheta)^\top \mDelta + \frac{c^2}{2}\,\mDelta^\top \nabla^2 f(\vtheta\pm\zeta^\pm c\mDelta)\,\mDelta.
\end{eqnarray}
Subtracting the two expansions gives
\begin{eqnarray}
\label{Taylor expansion of noiseless SPSA estimator2}
&\!\!\!\!& f(\vtheta + c\mDelta)-f(\vtheta - c\mDelta) = 2c \nabla f(\vtheta)^\top \mDelta  +  \frac{c^2}{2}\,\mDelta^\top\big[\nabla^2 f(\vtheta+\zeta^+c\mDelta)-\nabla^2 f(\vtheta-\zeta^-c\mDelta)\big]\mDelta\nonumber\\
&\!\!\Longrightarrow\!\!& \frac{f(\vtheta + c\mDelta)-f(\vtheta - c\mDelta)}{2c\Delta_\ell} =  \frac{\nabla f(\vtheta)^\top \mDelta}{\Delta_\ell}  +  \rho_\ell(\vtheta),
\end{eqnarray}
where we define $\rho_\ell(\vtheta) = \frac{c\mDelta^\top\big[\nabla^2 f(\vtheta+\zeta^+c\mDelta)-\nabla^2 f(\vtheta-\zeta^-c\mDelta)\big]\mDelta}{4\Delta_\ell}$.

Next, using the identity $\Delta_\ell^{-1}=\Delta_\ell$ for Rademacher perturbations, we obtain
\begin{eqnarray}
\label{Taylor expansion of noiseless SPSA estimator3}
&\!\!\!\!& \frac{\nabla f(\vtheta)^\top \mDelta}{\Delta_\ell}  = \sum_{j=1}^N \partial_j f(\vtheta)\frac{\Delta_j}{\Delta_\ell}\nonumber\\
&\!\!\Longrightarrow\!\!& \frac{\nabla f(\vtheta)^\top \mDelta}{\Delta_\ell}  = \sum_{j=1}^N \partial_j f(\vtheta)\Delta_j\Delta_\ell\nonumber\\
&\!\!\Longrightarrow\!\!& \frac{\nabla f(\vtheta)^\top \mDelta}{\Delta_\ell}  =  \partial_\ell f(\vtheta) + \sum_{j \neq \ell }  \partial_j f(\vtheta)\Delta_j\Delta_\ell.
\end{eqnarray}
Combining \eqref{decomposition of SPSA estimator1}, \eqref{Taylor expansion of noiseless SPSA estimator2}, and \eqref{Taylor expansion of noiseless SPSA estimator3}
gives
\begin{eqnarray}
\label{decomposition of SPSA estimator2}
\hat g_\ell(\vtheta) = \partial_\ell f(\vtheta) + \sum_{j \neq \ell }  \partial_j f(\vtheta)\Delta_j\Delta_\ell + \rho_\ell(\vtheta) + \delta_{\ell}(\vtheta),
\end{eqnarray}
where
\begin{eqnarray}
\label{definition of delta}
\delta_{\ell}(\vtheta) = \frac{\xi(\vtheta^+) - \xi(\vtheta^-)}{2c\Delta_\ell},
\end{eqnarray}
and
\begin{eqnarray}
\label{definition of rho}
\rho_\ell(\vtheta) = \frac{c\mDelta^\top\big[\nabla^2 f(\vtheta+\zeta^+c\mDelta)-\nabla^2 f(\vtheta-\zeta^-c\mDelta)\big]\mDelta}{4\Delta_\ell}.
\end{eqnarray}

\medskip
\textbf{Proof of \eqref{bias of noisy gradient}.} Since the entries of $\mDelta$ are independent Rademacher random variables, we have
$\E[\Delta_j\Delta_\ell]=\E[\Delta_j]\E[\Delta_\ell]=0$ for all $j\neq\ell$. Moreover, by the tower property of conditional expectation and the measurement noise model in \Cref{lem:noise-model}, we have
\begin{eqnarray}
\label{expectation of delta1}
\E[\delta_{\ell}(\vtheta)] = \E\big[\E[\delta_{\ell}(\vtheta) \mid \mDelta] \big] = \E\big[\tfrac{1}{2c\Delta_\ell} \E[\xi(\vtheta^+) - \xi(\vtheta^-) \mid \mDelta ]  \big] = 0,
\end{eqnarray}
where the second equality follows from the fact that $\Delta_\ell$ is measurable with respect to $\sigma(\mDelta)$, and the last equality follows from $\E[\xi(\vtheta^s)\mid\mDelta]=0$.

Furthermore, according to Proposition~\ref{Proposition: Hessian Property}, we obtain
\begin{eqnarray}
\label{upper bound of rho}
|\rho_\ell(\vtheta)| &\!\!\leq \!\!& \frac{c}{4} ||\mDelta ||_2^2\|\nabla^2 f(\vtheta+\zeta^+c\mDelta)-\nabla^2 f(\vtheta-\zeta^-c\mDelta) \|\nonumber\\
&\!\!\leq \!\!&   2c\,\|\vw\|_1\,N^{\frac{5}{2}}\|(\zeta^+ + \zeta^-   )c\mDelta\|_2\nonumber\\
&\!\!\leq \!\!&  4c^2 N^3\|\vw\|_1.
\end{eqnarray}

Taking expectations on both sides of \eqref{decomposition of SPSA estimator2} and using the above identities yields
\begin{eqnarray}
\label{bias of noisy gradient appendix}
\big|\E[\hat g_\ell(\vtheta)] - \partial_\ell f(\vtheta)\big| = |\E[\rho_\ell(\vtheta)]| \leq \E[|\rho_\ell(\vtheta)|] \leq 4c^2 N^3\|\vw\|_1.
\end{eqnarray}

\medskip
\textbf{Proof of \eqref{variance of noisy gradient}.} To analyze the variance of the SPSA estimator, we decompose \eqref{decomposition of SPSA estimator2} into a noise-free component and a measurement noise component:
\begin{eqnarray}
\label{decomposition of SPSA estimator another}
\hat g_\ell(\vtheta) &\!\! = \!\!&  \partial_\ell f(\vtheta) + \sum_{j \neq \ell }  \partial_j f(\vtheta)\Delta_j\Delta_\ell+ \rho_\ell(\vtheta)  + \delta_{\ell}(\vtheta)\nonumber\\
&\!\!:= \!\!& \wt g_\ell(\vtheta) + \delta_{\ell}(\vtheta),
\end{eqnarray}
where $\wt g_\ell(\vtheta)$ is measurable with respect to $\sigma(\mDelta)$. Since
$\E[\delta_{\ell}(\vtheta)\mid\mDelta]=0$, it follows from the law of total expectation that
\begin{eqnarray}
\label{expectation of cross}
\E[\wt g_\ell(\vtheta)\delta_{\ell}(\vtheta) ] = \E[\wt g_\ell(\vtheta)  \E[\delta_{\ell}(\vtheta) \mid \mDelta  ]] =0.
\end{eqnarray}
Consequently, the covariance between $\wt g_\ell(\vtheta)$ and $\delta_{\ell}(\vtheta)$ vanishes, yielding
\begin{eqnarray}
\label{covariance of g ll}
\Var(\hat g_\ell(\vtheta))= \Var(\wt g_\ell(\vtheta)) + \Var(\delta_{\ell}(\vtheta)).
\end{eqnarray}

First, we derive an upper bound for $\Var(\delta_{\ell}(\vtheta))$. By the conditional independence of $\xi(\vtheta^{+})$ and $\xi(\vtheta^{-})$ given $\mDelta$,
\begin{eqnarray}
\label{expansion of conditional noisy variance}
\Var(\delta_{\ell}(\vtheta) \mid \mDelta) = \frac{1}{4c^2} \Var(\xi(\vtheta^{+}) - \xi(\vtheta^{-}) \mid \mDelta) = \frac{1}{4c^2} \Big( \Var(\xi(\vtheta^{+}) \mid \mDelta)
+ \Var(\xi(\vtheta^{-}) \mid \mDelta) \Big).
\end{eqnarray}
Applying \Cref{lem:noise-model} yields
\begin{eqnarray}
\label{upper bound for conditional noisy variance1}
\Var(\xi(\vtheta^{\pm}) \mid \mDelta) \leq \frac{\|\vw\|_2^2}{M},
\end{eqnarray}
and therefore
\begin{eqnarray}
\label{upper bound for conditional noisy variance2}
\Var(\delta_{\ell}(\vtheta) \mid \mDelta) \leq \frac{\|\vw\|_2^2}{2c^2M}.
\end{eqnarray}
Taking expectation and using the law of total variance, we have
\begin{eqnarray}
\label{upper bound for conditional noisy variance3}
\Var(\delta_{\ell}(\vtheta)) = \E[\Var(\delta_{\ell}(\vtheta) \mid \mDelta)] + \Var(\E[\delta_{\ell}(\vtheta) \mid \mDelta]).
\end{eqnarray}
Since $\E[\delta_{\ell}(\vtheta) \mid \mDelta]=0$, the second term vanishes, and hence
\begin{eqnarray}
\label{upper bound for conditional noisy variance4}
\Var(\delta_{\ell}(\vtheta)) = \E[\Var(\delta_{\ell}(\vtheta) \mid \mDelta)] \leq \frac{\|\vw\|_2^2}{2c^2M}.
\end{eqnarray}

Next, we derive an upper bound for $\Var(\wt g_\ell(\vtheta))$. Recall that $\wt g_\ell(\vtheta) = \partial_\ell f(\vtheta) + \sum_{j \neq \ell }  \partial_j f(\vtheta)\Delta_j\Delta_\ell + \rho_\ell(\vtheta)   := g_\ell^0(\vtheta) + \rho_\ell(\vtheta)$. Expanding the variance gives
\begin{eqnarray}
\label{expansion of non free parameter}
\Var(\wt g_\ell(\vtheta)) = \Var(g_\ell^0(\vtheta)) + \Var(\rho_\ell(\vtheta)) + 2 \Cov(g_\ell^0(\vtheta), \rho_\ell(\vtheta) ).
\end{eqnarray}
By \eqref{upper bound of rho}, we have
\begin{eqnarray}
\label{expansion of non free parameter1}
\Var(\rho_\ell(\vtheta)) \leq \E[ (\rho_\ell(\vtheta))^2 ] \leq 16 c^4 N^6 \|\vw\|_1^2.
\end{eqnarray}
We next analyze the variance of $g_\ell^0(\vtheta)$. Since $\E[g_\ell^0(\vtheta)] = \E[\partial_\ell f(\vtheta)] + \E[\sum_{j \neq \ell }  \partial_j f(\vtheta)\Delta_j] \E[\Delta_\ell] = \partial_\ell f(\vtheta)$ using $\E[\Delta_\ell] = 0$, it follows that
\begin{eqnarray}
\label{expansion of non free parameter variance}
\Var(g_\ell^0(\vtheta)) &\!\! =\!\!& \E[(g_\ell^0(\vtheta) -  \E[g_\ell^0(\vtheta)] )^2 ] = \E\bigg[\bigg(\sum_{j \neq \ell }  \partial_j f(\vtheta)\Delta_j\Delta_\ell\bigg)^2 \bigg]\nonumber\\
&\!\! = \!\!&\E\bigg[\bigg(\sum_{j \neq \ell }  \partial_j f(\vtheta)\Delta_j \bigg)^2 \bigg]\nonumber\\
&\!\! = \!\!& \sum_{j \neq \ell } (\partial_j f(\vtheta))^2 \E[\Delta_j^2 ] + \sum_{j\neq k \atop j,k\neq \ell} \partial_j f(\vtheta) \partial_k f(\vtheta) \E[\Delta_j\Delta_k]\nonumber\\
&\!\! = \!\!& \sum_{j \neq \ell } (\partial_j f(\vtheta))^2\nonumber\\
&\!\! = \!\!& \|\nabla f(\vtheta)\|_2^2 - (\partial_{\ell} f(\vtheta))^2.
\end{eqnarray}
Applying the Cauchy--Schwarz inequality together with \eqref{expansion of non free parameter1} and \eqref{expansion of non free parameter variance}, we obtain
\begin{eqnarray}
\label{upper bound of covariance term}
|\Cov(g_\ell^0(\vtheta), \rho_\ell(\vtheta) )| \leq \sqrt{\Var(\rho_\ell(\vtheta)) \Var(g_\ell^0(\vtheta)) }\leq 4c^2 N^{3} \| \vw\|_1 \sqrt{\|\nabla f(\vtheta)\|_2^2 - (\partial_{\ell} f(\vtheta))^2}.
\end{eqnarray}
Substituting \eqref{expansion of non free parameter1}, \eqref{expansion of non free parameter variance}, and \eqref{upper bound of covariance term} into \eqref{expansion of non free parameter} yields
\begin{eqnarray}
\label{expansion of non free parameter conclusion appendix}
\Var(\wt g_\ell(\vtheta)) \leq \|\nabla f(\vtheta)\|_2^2 - (\partial_{\ell} f(\vtheta))^2 + 16  c^4 N^6 \|\vw\|_1^2 +  8c^2 N^{3} \| \vw\|_1 \sqrt{\|\nabla f(\vtheta)\|_2^2 - (\partial_{\ell} f(\vtheta))^2}.
\end{eqnarray}
Finally, combining \eqref{covariance of g ll} with \eqref{upper bound for conditional noisy variance4} yields
\begin{eqnarray}
\label{covariance of g ll conclusion appendix}
\Var(\hat g_\ell(\vtheta)) \leq \|\nabla f(\vtheta)\|_2^2 - (\partial_{\ell} f(\vtheta))^2 + \frac{\|\vw\|_2^2}{2c^2M} + 16  c^4 N^6 \|\vw\|_1^2 +  8c^2 N^{3} \| \vw\|_1 \sqrt{\|\nabla f(\vtheta)\|_2^2 - (\partial_{\ell} f(\vtheta))^2}.
\end{eqnarray}
By using $\Var(\hat g_\ell(\vtheta) - \partial_\ell f(\vtheta)) = \Var(\hat g_\ell(\vtheta))$, this completes the proof.

\end{proof}

\section{Proof of Corollary~\ref{cor:sample complexity barren plateau}}
\label{Proof of SNR lower bound requirement}

\begin{proof}
In the following, we use ${\mathrm{SNR}}_\ell := \frac{ \E_{\Theta} [(\partial_\ell f(\vtheta))^2] }{ \E_{\Theta} [\mathrm{MSE}_\ell(\vtheta)] }$ to characterize the measurement budget required for reliable SPSA optimization. By \Cref{thm:main property of noisy gradient}, together with \eqref{eq:MSE} and the bias bound \eqref{bias of noisy gradient}, we have
\begin{eqnarray}
\label{population SNR lower bound}
{\mathrm{SNR}}_\ell \geq \frac{ \E_{\Theta} [(\partial_\ell f(\vtheta))^2] }{L_d}.
\end{eqnarray}
where $L_d = \E_{\Theta} [\|\nabla f(\vtheta)\|_2^2 - (\partial_{\ell} f(\vtheta))^2] + \frac{\|\vw\|_2^2}{2c^2M} + 32  c^4 N^6 \|\vw\|_1^2 +  8c^2 N^{3} \| \vw\|_1 \E_{\Theta}\Big[\sqrt{\|\nabla f(\vtheta)\|_2^2 - (\partial_{\ell} f(\vtheta))^2}\,\Big]$.

Let $b:=\E_{\Theta}[(\partial_\ell f(\vtheta))^2]$. Since $\E_{\Theta}[\|\nabla f(\vtheta)\|_2^2] = \sum_{\ell=1}^{N}\E_{\Theta}[(\partial_\ell f(\vtheta))^2]$, we adopt the scaling $\E_{\Theta}[\|\nabla f(\vtheta)\|_2^2] = N b$, so that
\begin{eqnarray}
\E_{\Theta}[\|\nabla f(\vtheta)\|_2^2 - (\partial_\ell f(\vtheta))^2] = (N-1)b,
\end{eqnarray}
and, by Jensen's inequality,
\begin{eqnarray}
\E_{\Theta}\Big[\sqrt{\|\nabla f(\vtheta)\|_2^2 - (\partial_{\ell} f(\vtheta))^2}\,\Big] \;\le\; \sqrt{(N-1)b}.
\end{eqnarray}
To characterize the scaling behavior, we choose the perturbation parameter $c$ such that the finite-difference bias is commensurate with the intrinsic gradient scale, namely
\begin{eqnarray}
\label{c scaling}
c^2 \;=\; \gamma\sqrt{b} \;=\; \Theta\Big(\sqrt{\E_{\Theta}[(\partial_\ell f(\vtheta))^2]}\Big)
\end{eqnarray}
for a constant $\gamma>0$. Under \eqref{c scaling}, the three bias-related terms in the denominator of \eqref{population SNR lower bound} all scale linearly in $b$:
\begin{eqnarray}
&\!\!\!\!&(N-1)b + 32c^4N^6\|\vw\|_1^2 + 8c^2N^3\|\vw\|_1\sqrt{(N-1)b}\nonumber\\
&\!\!=\!\!& b\Big[(N-1) + 32\gamma^2N^6\|\vw\|_1^2 + 8\gamma N^3\sqrt{N-1}\,\|\vw\|_1\Big]\nonumber\\
&\!\!\le\!\!& b\,A(N,\vw,\gamma),
\end{eqnarray}
where we have introduced the constant
\begin{eqnarray}
\label{definition of A}
A(N,\vw,\gamma) \;:=\; \Big(\sqrt{N-1} + 4\sqrt{2}\,\gamma N^3\|\vw\|_1\Big)^{2}.
\end{eqnarray}
The SNR lower bound therefore becomes
\begin{eqnarray}
\label{population SNR lower bound simplified}
{\mathrm{SNR}}_\ell \;\geq\; \frac{ b }{ A(N,\vw,\gamma)\,b + \dfrac{\|\vw\|_2^2}{2\gamma\sqrt{b}\,M} }.
\end{eqnarray}
To ensure $\mathrm{SNR}_\ell \ge \varepsilon_1/A(N,\vw,\gamma)$ for a prescribed $\varepsilon_1\in(0,1)$, it suffices to impose $\varepsilon_1/A(N,\vw,\gamma)$ as a lower bound on the right-hand side of \eqref{population SNR lower bound simplified}, which gives
\begin{eqnarray}
\label{snr requirement expanded}
(1-\varepsilon_1)\,b \;\ge\; \frac{\varepsilon_1\|\vw\|_2^2}{2\gamma\,A(N,\vw,\gamma)\sqrt{b}\,M}.
\end{eqnarray}
Since $\varepsilon_1<1$, the left-hand side of \eqref{snr requirement expanded} is strictly positive, and rearranging \eqref{snr requirement expanded} gives
\begin{eqnarray}
\label{sample complexity explicit}
M \;\ge\; \frac{\varepsilon_1}{1-\varepsilon_1}\cdot\frac{\|\vw\|_2^2}{2\gamma\,A(N,\vw,\gamma)\,b^{3/2}},
\end{eqnarray}
which completes the proof.
\end{proof}

\section{Proof of \Cref{thm: upper bound of the gradient energy}}
\label{Proof of upper bound for gradient}

\begin{proof}
In this proof, we set
\begin{eqnarray}
\label{constants setting}
\mu_t=\mu_0T^{-1/2},\qquad
c_t=c_0T^{-1/8},\qquad
M_t=M_0T^{1/4},
\end{eqnarray}
where
\begin{eqnarray}
\label{constants setting requirements}
\mu_0\le \frac{1}{4(3N^2+N)\|\vw\|_1},\qquad
c_0>0,\qquad
M_0>0.
\end{eqnarray}
Since these schedules are constant with respect to the iteration index $t$, we have
\begin{eqnarray}
\label{summation of step sizes}
\sum_{t=0}^{T-1}\mu_t = \mu_0 T^{\frac{1}{2}}.
\end{eqnarray}

Applying \Cref{lemma:convergence of T time} yields
\begin{eqnarray}
\label{expansion of iteration2 another4 upper bound energy1}
\min_{0\leq t \leq T-1}\E[\|\nabla f(\vtheta_t)\|_2^2] \leq \frac{4 (f(\vtheta_0) - f^\star)}{\sum_{t=0}^{T-1}\mu_t} + \frac{4N^2\|\vw\|_1\|\vw\|_2^2\sum_{t=0}^{T-1}\frac{\mu_t^2}{c_t^2 M_t}}{\sum_{t=0}^{T-1}\mu_t}  + \frac{72N^{7}\|\vw\|_1^2\sum_{t=0}^{T-1}c_t^4\mu_t}{\sum_{t=0}^{T-1}\mu_t}.
\end{eqnarray}
By \Cref{lem:fbounded}, we have
\begin{eqnarray}
\label{difference upper bound 1}
f(\vtheta_0)-f^\star \le 2\|\vw\|_1.
\end{eqnarray}
Moreover, we can obtain
\begin{eqnarray}
\label{summation transform1}
\sum_{t=0}^{T-1}\frac{\mu_t^2}{c_t^2 M_t} = \frac{\mu_0^2}{c_0^2 M_0},
\end{eqnarray}
and
\begin{eqnarray}
\label{summation transform2}
\sum_{t=0}^{T-1}c_t^4\mu_t = c_0^4 \mu_0.
\end{eqnarray}
Substituting the above identities together with \eqref{summation of step sizes} into \eqref{expansion of iteration2 another4 upper bound energy1}
gives
\begin{eqnarray}
\label{expansion of iteration2 another4 upper bound energy2}
\min_{0\leq t \leq T-1}\E[\|\nabla f(\vtheta_t)\|_2^2] \leq \kappa\, T^{-\frac{1}{2}},
\end{eqnarray}
where
\begin{eqnarray}
\label{definition of kappa}
\kappa := \frac{8 \|\vw\|_1}{\mu_0} + \frac{4N^2\|\vw\|_1\|\vw\|_2^2 \mu_0}{c_0^2 M_0} + 72 N^{7}\|\vw\|_1^2 c_0^4.
\end{eqnarray}
To guarantee $\min_{0\leq t \leq T-1}\E[\|\nabla f(\vtheta_t)\|_2^2] \leq \epsilon_2\, \E_{\Theta}[\|\nabla f(\vtheta)\|_2^2]$, it suffices to choose the number of iterations as
\begin{eqnarray}
\label{the necessary condition for the number of iterations}
T \;\ge\; \frac{\kappa^2}{\epsilon_2^2\, \big(\E_{\Theta}[\|\nabla f(\vtheta)\|_2^2]\big)^2 }.
\end{eqnarray}
With $T$ chosen as the smallest integer satisfying \eqref{the necessary condition for the number of iterations}, the total measurement budget of the SPSA estimator is
\begin{eqnarray}
\label{the necessary condition for the number of measurements shots for SPSA}
N_{\textup{SPSA}} = 2 L M_t T = 2L M_0 T^{\frac54} = O\bigg(\frac{L M_0\kappa^{\frac{5}{2}}}{ \epsilon_2^{\frac{5}{2}} \big(\E_{\Theta}[\|\nabla f(\vtheta)\|_2^2]\big)^{\frac{5}{2}}} \bigg) .
\end{eqnarray}
\end{proof}

\section{Auxiliary Materials}

\begin{lemma}
\label{lem:fbounded}
Consider the VQE objective $f(\vtheta) = \langle \vphi_0| \mU^\dagger(\vtheta)\mH\mU(\vtheta) |\vphi_0\rangle$ with Hamiltonian decomposition $\mH=\sum_{\alpha=1}^{L}w_\alpha\mP_\alpha$, and denote  $\vw = \begin{bmatrix} w_1 & \cdots & w_L \end{bmatrix}^\top\in\R^{L}$. Then, for any parameter vector $\vtheta$,
\begin{eqnarray}
\label{bound of f}
-\|\vw\|_1 \le f(\vtheta) \le \|\vw\|_1.
\end{eqnarray}
As a consequence, the optimal value $f^\star :=\inf_{\vtheta} f(\vtheta)$ satisfies $f^\star\geq -\|\vw\|_1$. Moreover, the optimality gap is uniformly bounded as
\begin{eqnarray}
\label{bound of difference f and f star}
f(\vtheta)-f^\star \le 2\|\vw\|_1 .
\end{eqnarray}
\end{lemma}

\begin{proof}
By the triangle inequality of the spectral norm, we have
\begin{eqnarray}
\label{upper bound of Hal H}
\|\mH\| \leq
\sum_{\alpha=1}^{L}|w_\alpha|\|\mP_\alpha\| = \|\vw\|_1,
\end{eqnarray}
where we use the fact that $\|\mP_\alpha\|=1$ for Pauli operators. Since $\mH$ is Hermitian, all eigenvalues of $\mH$ lie in the interval $[-\|\mH\|,\|\mH\|]$, which is contained in $[-\|\vw\|_1,\|\vw\|_1]$.

Moreover, $f(\vtheta)$ is the Rayleigh quotient of $\mH$ evaluated at the normalized vector $\mU(\vtheta)|\vphi_0\rangle$. Therefore, it is bounded by the extreme eigenvalues of $\mH$, yielding
\begin{eqnarray}
\label{bound of f proof}
-\|\vw\|_1\leq f(\vtheta)\leq \|\vw\|_1 .
\end{eqnarray}
Taking the infimum over $\vtheta$ gives
\begin{eqnarray}
\label{lower bound of f proof}
f^\star=\inf_{\vtheta}f(\vtheta)\geq -\|\vw\|_1 .
\end{eqnarray}
Combining this lower bound with the upper bound of $f(\vtheta)$ completes the proof of
\begin{eqnarray}
\label{upper bound of difference f and f star proof}
f(\vtheta)-f^\star\leq 2\|\vw\|_1 .
\end{eqnarray}
\end{proof}

\begin{lemma}
\label{lem:smooth}
Under the setting of Proposition~\ref{Proposition: Hessian Property}, the gradient of the VQE objective function $f$ is globally Lipschitz continuous. Specifically, for any $\vx,\vy\in\R^{N}$,
\begin{eqnarray}
\label{Lipschitz constant for function f}
\|\nabla f(\vx)-\nabla f(\vy)\|_2 \leq 4N\|\vw\|_1\|\vx-\vy\|_2 .
\end{eqnarray}
\end{lemma}

\begin{proof}
By \eqref{second derivation of f1 sorted}, the second-order partial derivative of each Pauli term satisfies
\begin{eqnarray}
\label{second derivation of f1 another proof}
\partial_{\ell j} f_\alpha(\vtheta) = -\langle\vphi_0|\mU^\dagger(\vtheta) [\wt \mG_{j'}(\vtheta),  [\wt \mG_{\ell'}(\vtheta), \mP_\alpha] ] \mU(\vtheta)|\vphi_0\rangle,\quad \ell,j\in\{1,\ldots,N\},
\end{eqnarray}
where $\ell'=\max\{\ell,j\}$ and $j'=\min\{\ell,j\}$. Moreover, since $\|\wt \mG_\ell(\vtheta)\|\leq 1$ and $\|\mP_\alpha\|=1$, the nested commutator satisfies $\|[\wt \mG_{j'}(\vtheta),  [\wt \mG_{\ell'}(\vtheta), \mP_\alpha] ]\|\leq 4$. Therefore,
\begin{eqnarray}
\label{upper bound second derivation of f1 another proof}
|\partial_{\ell j} f(\vtheta)| = \bigg|\sum_{\alpha=1}^{L}w_\alpha \partial_{\ell j} f_\alpha(\vtheta) \bigg| \leq  \sum_{\alpha=1}^{L} |w_\alpha| |\partial_{\ell j} f_\alpha(\vtheta)| \leq 4\|\vw\|_1.
\end{eqnarray}
By Clairaut's theorem, the Hessian $\nabla^2 f(\vtheta)$ is symmetric. Hence, its spectral norm can be bounded by the maximum absolute row sum, which gives
\begin{eqnarray}
\label{upper bound spectral norm second derivation of f1 another proof}
\|\nabla^2f(\vtheta)\| \leq \max_{\ell} \sum_{j=1}^{N}|\partial_{\ell j} f(\vtheta)| \leq 4N\|\vw\|_1.
\end{eqnarray}
Finally, applying the fundamental theorem of calculus along the line segment connecting $\vx$ and $\vy$, we obtain
\begin{eqnarray}
\label{fundamental theorem of calculus another proof}
\nabla f(\vx)-\nabla f(\vy) = \int_0^1 \nabla^2f\bigl(\vy+t(\vx-\vy)\bigr)(\vx-\vy)\,dt.
\end{eqnarray}
Taking norms on both sides yields $\|\nabla f(\vx)-\nabla f(\vy)\|_2\le\int_0^1\|\nabla^2f\bigl(\vy+t(\vx-\vy)\bigr)\|\,dt\,\|\vx-\vy\|_2\le 4N\|\vw\|_1 \|\vx-\vy\|_2$.
\end{proof}

\begin{lemma}
\label{lem:onestep of SPSA}
Under the noise model in \Cref{lem:noise-model}, suppose that $\mu_t \leq \frac{1}{4(3N^2+N)\|\vw\|_1}$. Let $\{\calF_t\}_{t\ge0}$ denote the natural filtration generated by the iterates
$\{\vtheta_t\}_{t\ge0}$, i.e., $\calF_t:=\sigma(\vtheta_0,\dots,\vtheta_t)$.  Since the perturbation vectors and measurement noises are independently generated at each iteration, the conditional expectation satisfies $\E[\cdot\mid\calF_t]=\E[\cdot\mid\vtheta_t]$. Then the following one-step descent inequality holds:
\begin{eqnarray}
\label{upper bound noiseless gradient in expectation}
\E[f(\vtheta_{t+1})\mid \calF_t  ] \leq   f(\vtheta_t) - \frac{\mu_t}{4}\|\nabla f(\vtheta_t)\|_2^2 + \frac{N^2\|\vw\|_1\|\vw\|_2^2}{c_t^2 M} \mu_t^2 + 18c_t^4 N^{7}\|\vw\|_1^2\mu_t.
\end{eqnarray}
The conditional expectation is taken with respect to the randomness of the SPSA perturbation $\mDelta_t=[\Delta_{1,t},\ldots,\Delta_{N,t}]^\top$ and the measurement noise $\xi(\vtheta_t)$ at iteration $t$.
\end{lemma}

\begin{proof}
By \Cref{lem:smooth}, the gradient $\nabla f$ is $4N\|\vw\|_1$-Lipschitz. Therefore, the standard descent lemma implies
\begin{eqnarray}
\label{expansion of iteration}
f(\vtheta_{t+1}) \le f(\vtheta_t) - \mu_t\langle\nabla f(\vtheta_t),\hat g(\vtheta_t)\rangle + 2N\|\vw\|_1\mu_t^2\|\hat g(\vtheta_t)\|_2^2 .
\end{eqnarray}
Taking conditional expectation with respect to $\calF_t$, we first bound the linear term. We have
\begin{eqnarray}
\label{conditional expectation on the cross term}
\E[\langle\nabla f(\vtheta_t),\hat g(\vtheta_t)\rangle \mid \calF_t] &\!\!=\!\!& \|\nabla f(\vtheta_t) \|_2^2 + \<\nabla f(\vtheta_t), \E[ \hat g(\vtheta_t) \mid \calF_t ]- \nabla f(\vtheta_t) \>\nonumber\\
&\!\!\geq \!\!& \frac{1}{2}\|\nabla f(\vtheta_t) \|_2^2 - \frac{1}{2}\| \E[ \hat g(\vtheta_t) \mid \calF_t ] - \nabla f(\vtheta_t)\|_2^2,
\end{eqnarray}
where the inequality follows from Young's inequality.
It remains to bound $\E[\|\hat g(\vtheta_t)\|_2^2\mid\calF_t]$. Using the variance decomposition, we have
\begin{eqnarray}
\label{expansion of squared term}
\E[\|\hat g(\vtheta_t)\|_2^2 \mid \calF_t ] = \big\|\E[\hat g(\vtheta_t)\mid\calF_t]\big\|_2^2 + \sum_{\ell=1}^N \Var(\hat g_{\ell}(\vtheta_t)\mid\calF_t).
\end{eqnarray}
We bound the two terms on the right-hand side separately. By the bias estimate \eqref{bias of noisy gradient appendix},
\begin{eqnarray}
\label{upper bound of the first term in squared term bound}
\| \E[ \hat g(\vtheta_t) \mid \calF_t ] - \nabla f(\vtheta_t)\|_2^2 \leq 16c_t^4 N^{7}\|\vw\|_1^2.
\end{eqnarray}
Consequently, we can obtain
\begin{eqnarray}
\label{upper bound of the first term in squared term}
\big\|\E[\hat g(\vtheta_t)\mid\calF_t]\big\|_2^2 &\!\!\leq \!\!& 2\|\nabla f(\vtheta_t) \|_2^2 + 2\| \E[ \hat g(\vtheta_t) \mid \calF_t ] - \nabla f(\vtheta_t)\|_2^2\nonumber\\
&\!\!\leq \!\!& 2\|\nabla f(\vtheta_t) \|_2^2 + 32c_t^4 N^{7}\|\vw\|_1^2.
\end{eqnarray}
Moreover, applying \eqref{covariance of g ll conclusion appendix} gives
\begin{eqnarray}
\label{covariance of g ll conclusion appendix another}
&\!\!\!\!&\Var(\hat g_\ell(\vtheta_t)\mid\calF_t)\nonumber\\
&\!\!\leq \!\!& \|\nabla f(\vtheta_t)\|_2^2 - (\partial_{\ell} f(\vtheta_t))^2 + \frac{\|\vw\|_2^2}{2c_t^2M} + 16  c_t^4 N^6 \|\vw\|_1^2 +  8c_t^2 N^{3} \| \vw\|_1 \sqrt{\|\nabla f(\vtheta_t)\|_2^2 - (\partial_{\ell} f(\vtheta_t))^2}\nonumber\\
&\!\!\leq \!\!& \frac{3}{2}\big(\|\nabla f(\vtheta_t)\|_2^2 - (\partial_{\ell} f(\vtheta_t))^2 \big) + \frac{\|\vw\|_2^2}{2c_t^2M} + 48  c_t^4 N^6 \|\vw\|_1^2,
\end{eqnarray}
where the last inequality uses $8xy\le \frac{1}{2}x^2+32y^2$ with $x=\sqrt{\|\nabla f(\vtheta_t)\|_2^2 - (\partial_{\ell} f(\vtheta_t))^2}$ and $y=c_t^2N^3\|\vw\|_1$.
Summing the above inequality over $\ell=1,\ldots,N$ yields
\begin{eqnarray}
\label{summation covariance of g ll conclusion appendix another}
\sum_{\ell=1}^N\Var(\hat g_\ell(\vtheta_t)\mid\calF_t) \leq   \frac{3(N-1)}{2}\|\nabla f(\vtheta_t)\|_2^2  + \frac{N\|\vw\|_2^2}{2c_t^2M} + 48  c_t^4 N^7 \|\vw\|_1^2.
\end{eqnarray}
Combining \eqref{upper bound of the first term in squared term} and \eqref{summation covariance of g ll conclusion appendix another}, we obtain
\begin{eqnarray}
\label{expansion of squared term another}
\E[\|\hat g(\vtheta_t)\|_2^2 \mid \calF_t ] \leq \frac{3N+1}{2}\|\nabla f(\vtheta_t)\|_2^2  + \frac{N\|\vw\|_2^2}{2c_t^2M} + 80  c_t^4 N^7 \|\vw\|_1^2.
\end{eqnarray}
Substituting \eqref{conditional expectation on the cross term}, \eqref{upper bound of the first term in squared term bound}, and \eqref{expansion of squared term another} into \eqref{expansion of iteration} gives
\begin{eqnarray}
\label{expansion of iteration1}
\E[f(\vtheta_{t+1})\mid \calF_t  ] &\!\!\leq \!\!& f(\vtheta_t) - \frac{\mu_t}{2}(1 - 2(3N^2+N)\|\vw\|_1 \mu_t )\|\nabla f(\vtheta_t)\|_2^2\nonumber\\
&\!\! \!\!& + \frac{N^2\|\vw\|_1\|\vw\|_2^2}{c_t^2 M} \mu_t^2 + 8c_t^4 N^{7}\|\vw\|_1^2 (\mu_t + 20N\|\vw\|_1\mu_t^2 ).
\end{eqnarray}
Finally, under the step size condition $\mu_t \leq \frac{1}{4(3N^2+N)\|\vw\|_1}$, we have $1-2(3N^2+N)\|\vw\|_1\mu_t\ge\frac{1}{2}$ and $20N\|\vw\|_1\mu_t\leq \frac{20N\|\vw\|_1}{4(3N^2+N)\|\vw\|_1}=\frac{5}{3N+1}\leq \frac{5}{4}$ for $N\geq 1$, so that $8(\mu_t+20N\|\vw\|_1\mu_t^2)\le 18\mu_t$. Therefore, \eqref{expansion of iteration1} reduces to
\begin{eqnarray}
\label{expansion of iteration2}
\E[f(\vtheta_{t+1})\mid \calF_t  ] \leq   f(\vtheta_t) - \frac{\mu_t}{4}\|\nabla f(\vtheta_t)\|_2^2 + \frac{N^2\|\vw\|_1\|\vw\|_2^2}{c_t^2 M} \mu_t^2 + 18c_t^4 N^{7}\|\vw\|_1^2\mu_t,
\end{eqnarray}
which completes the proof.
\end{proof}

\begin{lemma}
\label{lemma:convergence of T time}
Under the setting of \Cref{lem:onestep of SPSA}, suppose that the measurement budget  at iteration $t$ is $M_t$. Then, after $T$ iterations,
\begin{eqnarray}
\label{expansion of iteration2 another4 conlcusion}
\min_{0\leq t \leq T-1}\E[\|\nabla f(\vtheta_t)\|_2^2] \leq \frac{4 (f(\vtheta_0) - f^\star)}{\sum_{t=0}^{T-1}\mu_t} + \frac{4N^2\|\vw\|_1\|\vw\|_2^2\sum_{t=0}^{T-1}\frac{\mu_t^2}{c_t^2 M_t}}{\sum_{t=0}^{T-1}\mu_t}  + \frac{72N^{7}\|\vw\|_1^2\sum_{t=0}^{T-1}c_t^4\mu_t}{\sum_{t=0}^{T-1}\mu_t}.
\end{eqnarray}
Here, the expectation is taken over all randomness generated by the SPSA perturbations
$\{\mDelta_t\}_{t=0}^{T-1}$ and the measurement noises
$\{\xi(\vtheta_t)\}_{t=0}^{T-1}$ throughout the algorithm.
\end{lemma}

\begin{proof}
Applying \Cref{lem:onestep of SPSA} at iteration $t$ gives
\begin{eqnarray}
\label{expansion of iteration2 another}
\E[f(\vtheta_{t+1})\mid \calF_t  ] \leq   f(\vtheta_t) - \frac{\mu_t}{4}\|\nabla f(\vtheta_t)\|_2^2 + \frac{N^2\|\vw\|_1\|\vw\|_2^2}{c_t^2 M_t} \mu_t^2 + 18c_t^4 N^{7}\|\vw\|_1^2\mu_t.
\end{eqnarray}

Taking expectation over all randomness generated up to iteration $t$, namely
$\{\mDelta_s,\xi(\vtheta_s)\}_{s=0}^{t}$, and using the tower property
\begin{eqnarray}
\label{tower property at time t}
\E\!\left[\E[f(\vtheta_{t+1})\mid\calF_t]\right] = \E[f(\vtheta_{t+1})],
\end{eqnarray}
we obtain
\begin{eqnarray}
\label{expansion of iteration2 another1}
\E[f(\vtheta_{t+1})] \leq   \E[f(\vtheta_t)] - \frac{\mu_t}{4}\E[\|\nabla f(\vtheta_t)\|_2^2] + \frac{N^2\|\vw\|_1\|\vw\|_2^2}{c_t^2 M_t} \mu_t^2 + 18c_t^4 N^{7}\|\vw\|_1^2\mu_t.
\end{eqnarray}

Summing \eqref{expansion of iteration2 another1} over $t=0,\ldots,T-1$ yields
\begin{eqnarray}
\label{expansion of iteration2 another2}
&\!\! \!\!& \sum_{t=0}^{T-1}\E[f(\vtheta_{t+1})] \leq  \sum_{t=0}^{T-1}\E[f(\vtheta_t)] - \frac{1}{4}\sum_{t=0}^{T-1}\mu_t\E[\|\nabla f(\vtheta_t)\|_2^2] + N^2\|\vw\|_1\|\vw\|_2^2\sum_{t=0}^{T-1}\frac{\mu_t^2}{c_t^2 M_t}  + 18N^{7}\|\vw\|_1^2\sum_{t=0}^{T-1}c_t^4\mu_t\nonumber\\
&\!\!\Longrightarrow\!\!& \E[f(\vtheta_{T})] - f(\vtheta_0) \leq - \frac{1}{4}\sum_{t=0}^{T-1}\mu_t\E[\|\nabla f(\vtheta_t)\|_2^2] + N^2\|\vw\|_1\|\vw\|_2^2\sum_{t=0}^{T-1}\frac{\mu_t^2}{c_t^2 M_t}  + 18N^{7}\|\vw\|_1^2\sum_{t=0}^{T-1}c_t^4\mu_t,
\end{eqnarray}
where we used the fact that $\vtheta_0$ is deterministic, so that $\E[f(\vtheta_0)]=f(\vtheta_0)$.

Since
$f^\star=\inf_{\vtheta}f(\vtheta)$,
we have
$f(\vtheta_T)\ge f^\star$
almost surely.
Therefore,
\begin{eqnarray}
\label{lower bound of energy at time T}
\E[f(\vtheta_T)]\ge f^\star.
\end{eqnarray}
Substituting this bound into \eqref{expansion of iteration2 another2} gives
\begin{eqnarray}
\label{expansion of iteration2 another3}
&\!\! \!\!&\hspace{-0.5cm}f^\star - f(\vtheta_0) \leq - \frac{1}{4}\sum_{t=0}^{T-1}\mu_t\E[\|\nabla f(\vtheta_t)\|_2^2] + N^2\|\vw\|_1\|\vw\|_2^2\sum_{t=0}^{T-1}\frac{\mu_t^2}{c_t^2 M_t}  + 18N^{7}\|\vw\|_1^2\sum_{t=0}^{T-1}c_t^4\mu_t\nonumber\\
&\!\!\!\!&\hspace{-1cm}\Longrightarrow \sum_{t=0}^{T-1}\mu_t\E[\|\nabla f(\vtheta_t)\|_2^2] \leq 4 (f(\vtheta_0) - f^\star) + 4N^2\|\vw\|_1\|\vw\|_2^2\sum_{t=0}^{T-1}\frac{\mu_t^2}{c_t^2 M_t}  + 72N^{7}\|\vw\|_1^2\sum_{t=0}^{T-1}c_t^4\mu_t\nonumber\\
&\!\!\!\!&\hspace{-1cm}\Longrightarrow \frac{\sum_{t=0}^{T-1}\mu_t\E[\|\nabla f(\vtheta_t)\|_2^2]}{\sum_{t=0}^{T-1}\mu_t} \leq \frac{4 (f(\vtheta_0) - f^\star)}{\sum_{t=0}^{T-1}\mu_t} + \frac{4N^2\|\vw\|_1\|\vw\|_2^2\sum_{t=0}^{T-1}\frac{\mu_t^2}{c_t^2 M_t}}{\sum_{t=0}^{T-1}\mu_t}  + \frac{72N^{7}\|\vw\|_1^2\sum_{t=0}^{T-1}c_t^4\mu_t}{\sum_{t=0}^{T-1}\mu_t}.
\end{eqnarray}

Finally, since a weighted average is bounded below by the minimum of its entries,
\begin{eqnarray}
\label{inequality of minima value}
\frac{\sum_{t=0}^{T-1}\mu_t\E[\|\nabla f(\vtheta_t)\|_2^2]}{\sum_{t=0}^{T-1}\mu_t} \geq \min_{0\leq t \leq T-1}\E[\|\nabla f(\vtheta_t)\|_2^2],
\end{eqnarray}
the desired result follows immediately.
\end{proof}

\end{document}